\documentclass[runningheads]{llncs}

\usepackage[all,2cell]{xy}
\UseAllTwocells
\usepackage{amssymb}
\usepackage{proof}
\usepackage{color}
\usepackage[inline]{enumitem}
\usepackage{hyperref}
\usepackage{graphicx}

\usepackage{geometry}
\renewcommand{\arraystretch}{0.8}

\definecolor{dkblue}{rgb}{0,0.1,0.5}
\definecolor{dkgreen}{rgb}{0,0.4,0}
\definecolor{dkred}{rgb}{0.6,0,0}
\definecolor{dkpurple}{rgb}{0.7,0,1.0}
\definecolor{purple}{rgb}{0.9,0,1.0}
\definecolor{olive}{rgb}{0.4, 0.4, 0.0}
\definecolor{teal}{rgb}{0.0,0.4,0.4}
\definecolor{azure}{rgb}{0.0, 0.5, 1.0}
\definecolor{gray}{rgb}{0.5, 0.5, 0.5}
\definecolor{dkgrey}{rgb}{0.2, 0.2, 0.2}
\definecolor{lilac}{rgb}{0.70, 0.04, 0.08}
\newcommand{\comm}[3]{{\color{#1}[#2: #3]}}

\newcommand{\C}{\mathcal{C}}

\newcommand{\F}{\mathcal{F}}
\newcommand{\G}{\mathcal{G}}
\newcommand{\op}{^\mathrm{op}}
\newcommand{\rev}{^\mathrm{rev}}
\renewcommand{\d}[1]{#1^\circ}
\newcommand{\sd}[1]{#1^\bullet}
\newcommand{\Set}{\mathbf{Set}}

\newcommand{\CAT}{\mathbf{CAT}}
\newcommand{\Kl}{\mathbf{Kl}}
\newcommand{\id}{\mathsf{id}}
\newcommand{\comp}{\circ}
\newcommand{\sym}{\mathsf{sym}}
\newcommand{\ass}{\mathsf{ass}}
\newcommand{\Id}{\mathsf{Id}}
\newcommand{\ee}{\mathsf{e}}
\newcommand{\mm}{\mathsf{m}}
\newcommand{\phia}{\phi}
\newcommand{\phib}{\phi}
\newcommand{\psia}{\psi}
\newcommand{\psib}{\psi}
\newcommand{\eps}{\varepsilon}
\newcommand{\de}{\delta}
\newcommand{\Coalg}{\mathbf{Coalg}}
\newcommand{\Run}{\mathbf{Run}}
\newcommand{\Mnd}{\mathbf{Mnd}}
\newcommand{\Comnd}{\mathbf{Comnd}}

\newcommand{\Mon}{\mathbf{Mon}}
\newcommand{\Comon}{\mathbf{Comon}}
\newcommand{\Int}{\mathbf{IL}}
\newcommand{\mcInt}{\mathbf{MCIL}}
\newcommand{\fun}{\Rightarrow}

\newcommand{\out}{\mathsf{out}}
\newcommand{\ev}{\mathsf{ev}}

\newcommand{\fst}{\mathsf{fst}}
\newcommand{\snd}{\mathsf{snd}}
\newcommand{\zt}{{\ast}}
\newcommand{\pair}[2]{\langle #1, #2 \rangle}
\newcommand{\inl}{\mathsf{inl}}
\newcommand{\inr}{\mathsf{inr}}
\newcommand{\ldist}{\mathsf{ldist}}
\newcommand{\rdist}{\mathsf{rdist}}

\newcommand{\lollistar}{\mathrel{-\mkern-6mu\star}}
\newcommand{\dn}{\downarrow}
\newcommand{\Nat}{\mathbb{N}}
\newcommand{\Bool}{\mathbb{B}}
\newcommand{\btt}{\mathsf{tt}}
\newcommand{\bff}{\mathsf{ff}}
\newcommand{\bnot}{\mathsf{not}}
\newcommand{\Maybe}{\mathsf{Maybe}}
\newcommand{\just}{\mathsf{just}}
\newcommand{\nothing}{\mathsf{nothing}}
\newcommand{\dblt}{\mathsf{dblt}}
\newcommand{\St}{\mathsf{St}}
\newcommand{\Cost}{\mathsf{Cost}}
\newcommand{\Cont}{\mathsf{Cont}}

\newcommand{\letin}[2]{\mathsf{let~} #1 \mathsf{~in~} #2}

\newcommand{\pl}{\oplus}

\newcommand{\er}[1]{\comm{dkgreen}{ER}{#1}}
\newcommand{\tu}[1]{\comm{blue}{TU}{#1}}
\renewcommand{\er}[1]{}
\renewcommand{\tu}[1]{}

\renewcommand{\orcidID}[1]{\href{https://orcid.org/#1}{\includegraphics[height=9pt]{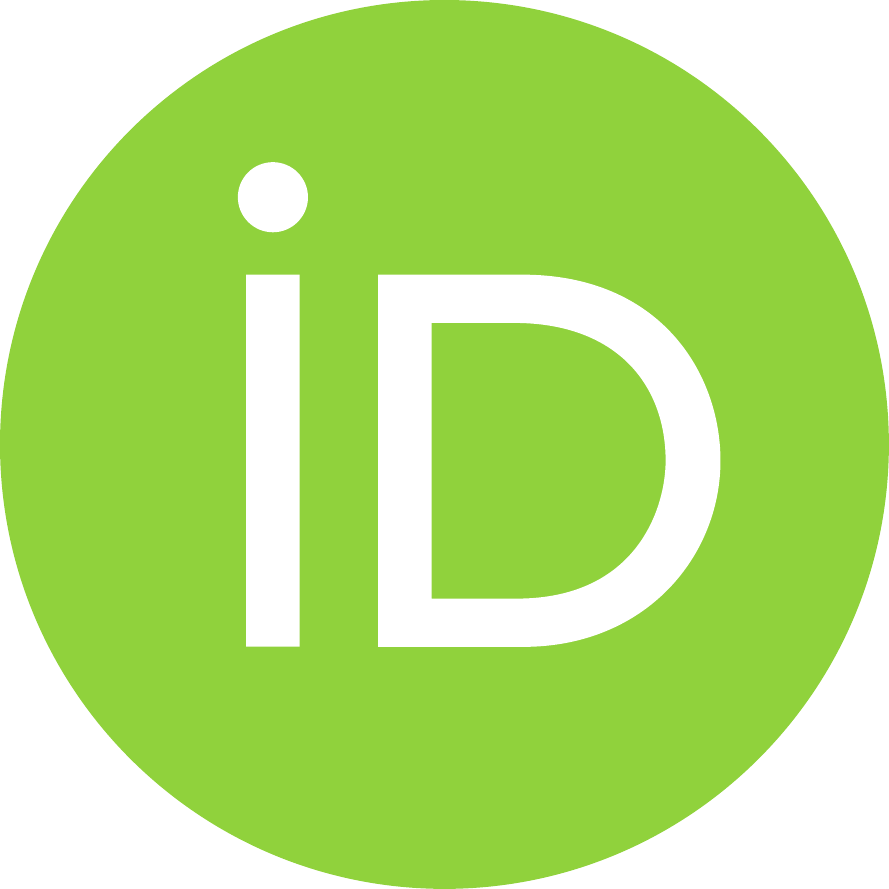}}}

\begin{document}



%
\title{Interaction Laws of Monads and Comonads}
%
%
\author{Shin-ya Katsumata\inst{1}\orcidID{0000-0001-7529-5489} \and
Exequiel Rivas\inst{2}\orcidID{0000-0002-2114-624X} \and
Tarmo Uustalu\inst{3,4}\orcidID{0000-0002-1297-0579}
}
\authorrunning{S.-y. Katsumata, E. Rivas and T. Uustalu}
%
\institute{National Institute of Informatics, Tokyo, Japan  
\and
Inria Paris, France 
\and
Dept.\ of Computer Science, Reykjavik University, Iceland 
\and
Dept.\ of Software Science, Tallinn University of Technology, Estonia \\
\email{s-katsumata@nii.ac.jp}, \email{exequiel.rivas-gadda@inria.fr}, 
\email{tarmo@ru.is}
}
\maketitle              
\begin{abstract}
  We introduce and study functor-functor and monad-comonad interaction
  laws as mathematical objects to describe interaction of
  effectful computations with behaviors of effect-performing machines. Monad-comonad
  interaction laws are monoid objects of the monoidal category of
  functor-functor interaction laws. We show that, for suitable
  generalizations of the concepts of dual and Sweedler dual, the
  greatest functor resp.\ monad interacting with a given functor or
  comonad is its dual while the greatest comonad interacting with a
  given monad is its Sweedler dual. We relate monad-comonad
  interaction laws to stateful runners. We show that functor-functor
  interaction laws are Chu spaces over the category of endofunctors
  taken with the Day convolution monoidal structure. Hasegawa's
  glueing endows the category of these Chu spaces with a monoidal
  structure whose monoid objects are monad-comonad
  interaction laws.
%
\end{abstract}


%
%
%
\section{Introduction}
\label{sec:introduction}



What does it mean to run an effectful program, abstracted into a
computation? 

In this paper, we take the view that an effectful computation does
not perform its effects; those are to be provided externally.
The computation can only proceed if placed in an environment that can provide its effects, e.g, respond to the computation's requests for input, listen to its
output, resolve its nondeterministic choices by tossing a coin,
consistently respond to its fetch and store commands. Abstractly, such an environment
is a machine whose implementation is opaque to
us; we can witness its behavior, its evolution through externally visible states. 

To formalize this intuition, we follow Moggi~\cite{Mog89} and Plotkin
and Power~\cite{PP03} in regards to allowed computations (the chosen
notions of computation) and describe them using a monad (resp.\
algebraic theory) $T$ on the category of types and functions that we
want to compute on. Allowed machine behaviors (the chosen notion of
machine behavior), at the same time, are described with a comonad
$D$. An operational semantics is then described by what we call an
interaction law, a natural transformation
$\psi : T X \times D Y \to X \times Y$ compatible with the (co)unit
and (co)multiplication.
This polymorphic function sends a computation $(TX)$ and a
machine behavior from some initial state ($DY$) into a return value
$X$ and a final state $Y$. It is also fine to
work with notions of computation and machine behavior that do not 
include ``just returning'' or/and are not closed under sequential composition; those can
be described with plain functors instead of a monad and a comonad. 

We take special interest in the questions (a) which is the ``greatest''
comonad interacting with the given monad $T$ (so any interaction
law of $T$ with any comonad would factor through the canonical
interaction law of $T$ with this comonad)? and (b) which is the ``greatest''
monad (resp.\ functor) interacting with a given comonad $D$ (or
functor $G$)? To answer these, we draw inspiration from algebra,
where the dual of a vector space $V$ is
$\d{V} = V \to \mathbb{K}$. The answer to (b) turns out to be: the dual of $D$ (resp.\ 
$G$), under a suitably generalized concept of dual. Question (a)
is harder. To answer it, we need to generalize the concept of 
what is called the Sweedler dual. The greatest comonad interacting with $T$ 
is the Sweedler dual of $T$. 

The contributions in this paper are the following:
\begin{enumerate}[label=(\roman*)]
\item We introduce \emph{functor-functor interaction laws},
  define
  the dual of a functor, and show that the greatest functor
  interacting with a given functor is its dual
  (Section~\ref{sec:functor-functor}).
\tu{I removed label (o) from here: I wanted it in other places,
as in the text we spoke of the first, second, third alternative
and I wanted this to agree with the number in the list, 
so the main option had to be (o)}
\item We study \emph{monad-comonad interaction laws} as monoid objects
  of the category of functor-functor interaction laws. We show that the dual lifts from functors to comonads 
and that the greatest
  monad interacting with a given comonad is its dual
  whereas for monads it does not lift like this; for the greatest comonad interacting with a monad, the   Sweedler dual is needed 
  (Section~\ref{sec:monad-comonad}).
\item We relate monad-comonad interaction laws to \emph{stateful runners}
  of Uustalu~\cite{Uus15} (Section~\ref{sec:running}).
\item Using the Day convolution and duoidal categories, we recast
  monad-comonad interaction laws as monoid-comonoid interaction laws,
  and relate them to two standard constructions: Chu spaces and \linebreak Hasegawa's glueing
  (Section~\ref{sec:monoid-comonoid}). This gives us a method 
  for computing the Sweedler duals of free monoids (monads) and their 
  quotients by equations. 
\end{enumerate}

We also introduce and study \emph{residual} functor-functor
interaction laws, monad-comonad interaction laws and stateful runners
as generalizations where the machine need not be able to perform all
effects of the computation (Section~\ref{sec:residual}).

\medskip

We assume the reader to be familiar with adjunctions/monads/comonads,
extensive categories \cite{CLW93}, Cartesian closed categories,
ends/coends (the end-coend calculus \cite{CW01,LO15}). In a nutshell,
extensive categories are categories with well-behaved finite
coproducts.

Throughout most of the paper
(Sections~\ref{sec:functor-functor}--\ref{sec:running}), we work with
one fixed base category $\C$ that we assume to be extensive with
finite products. For some constructions (the dual of a functor), we
also need that $\C$ is Cartesian closed. For the same constructions,
we also use certain ends
that we either explicitly show to exist or only use when 
they happen to exist. 
We also rely on Cartesian closedness in most examples.

\er{We don't have space, but we should describe the overview of the
  paper, at least right now the functor-functor interaction law
  appears as a surprise}

\section{Functor-functor interaction}
\label{sec:functor-functor}

We begin with functor-functor interaction, to then proceed to the
monad-comonad interaction laws in the next section.

\subsection{Functor-functor interaction laws}



In a functor-functor interaction law, computations over a set of
values $X$ are elements of $F X$ where $F$ is a given functor. Machine
behaviors over a set of states $Y$ are elements of $G Y$ where $G$ is
another given functor. Any allowed computation and any allowed machine
behavior can help each other reach a return value and a final state by
interacting as prescribed.

\medskip

We define an \emph{functor-functor interaction law} on $\C$ to be given by two
endofunctors $F$, $G$ together with a family of maps
\[
\phi_{X,Y} : F X \times G Y \to X \times Y
\]
natural in $X$ and $Y$.


\begin{example}
  The archetypical example of a functor-functor interaction law is
  defined by $F X = A \fun X$, $G Y = A \times Y$, and
  $\phi\, (f, (a, y)) = (f\, a, y)$ for some fixed object $A$. 
But we can also take, e.g., $F X = A \fun X$, $G Y = C \times Y$, and
$\phi\, (f, (c, y)) = (f\, (h\, c), y)$ for some fixed map
$h : C \to A$.
\end{example}

\begin{example} \label{ex:update-o}
  A more interesting example is obtained by taking
$F X = A \fun (B \times X)$, $G Y = A \times (B \fun Y)$,
$\phi\, (f, (a, g)) = \letin{(b, x) \leftarrow f\, a}{(x, g\, b)}$.
%
We can vary this by taking $G Y = (A \fun B) \fun (A \times Y)$ and
$\phi\, (f, h) = \letin{\pair{f_0}{f_1} \leftarrow f; (a, y) \leftarrow h\, f_0}{(f_1 a, y)}$.
\end{example}

%

\begin{example}
If $\C$ has the relevant
initial algebras and final coalgebras, we can get interaction laws by
iterating the above interactions, e.g., with
$F X = \mu Z.\, X + (A \fun (B \times Z))$ and
$G Y = \nu W.\, Y \times (A \times (B \fun W))$, or with
$F X = \nu Z.\, X + (A \fun (B \times Z))$ and
$G Y = \mu W.\, Y \times (A \times (B \fun W))$. We will shortly explain the
construction of $\phi$ in the first of these two cases.
\end{example}





An \emph{interaction law map} between $(F, G, \phi)$,
$(F', G', \phi')$ is given by natural transformations $f : F \to F'$,
$g : G' \to G$ such that 
$\phi_{X,Y} \comp (\id_{FX} \times g_Y) = \phi'_{X,Y} \comp (f_X \times \id_{G'Y})$.

Interaction laws form a category $\Int(\C)$, where the identity on
$(F, G, \phi)$ is $(\id_F, \id_G)$, and the composition of $(f, g) :
(F, G, \phi) \to (F', G', \phi')$ and $(f', g') : (F', G', \phi') \to
(F'', G'', \phi'')$ is $(f' \comp f, g \comp g')$. The condition on a
interaction law map is met for the composition because of the
commutation of the diagram
\[
\small
\xymatrix@R=1pc{
& & F X \times G Y \ar[r]^-{\phi_{X,Y}}
    & X \times Y \ar@{=}[dd] \\
& F X \times G' Y \ar[ur]^-{\id \times g_Y}  \ar[dr]_-{f_X \times \id} 
  & & \\
F X \times G'' Y \ar[ur]^-{\id \times g'_Y} \ar[dr]_-{f_X \times \id}
& & F' X \times G' Y \ar[r]^-{\phi'_{X,Y}}
    & X \times Y \ar@{=}[dd] \\
& F' X \times G'' Y \ar[ur]^-{\id \times g'_Y}  \ar[dr]_-{f'_X \times \id} 
  & & \\
& & F'' X \times G'' Y \ar[r]^-{\phi''_{X,Y}}
    & X \times Y 
}
\]

The composition monoidal structure of $[\C, \C]$ induces 
a similar monoidal structure on the category $\Int(\C)$.
The tensorial unit is $(\Id, \Id, \id_{\Id \times \Id})$. The tensor
of $(F, G, \phi)$ and $(J, K, \psi)$ is $(F \cdot J, G \cdot K, \psi
\comp \phi \cdot (J \times K))$. The tensor of $(f, g) : (F, G, \phi)
\to (F', G', \phi')$ and $(j, k) : (J, K, \psi) \to (J', K', \psi')$
is $(f \cdot j, g \cdot k)$. The condition on an interaction law map
is met by the commutation of
\[
\small
\xymatrix@R=0.7pc@C=0.4pc{
& & F (J X) \times G (K Y) \ar[rr]^{\phi_{J X, K Y}} 
    & & J X \times K Y \ar[rr]^-{\psi_{X,Y}}
        & & X \times Y \ar@{=}[dddd] \\
& F (J X) \times G (K' Y)  \ar[rr]^{\phi_{JX,K'Y}} \ar[ur]^-{\id \times Gk_Y}
  & & J X \times K' Y \ar[ur]_-{\id \times k_Y} \ar@{=}[dd]
\\
\hspace*{-1cm} F (J X) \times G' (K' Y) \hspace*{-1cm} \ar[ur]^{\id \times g_{KY}}  \ar[dr]_-{f_{JX} \times \id} 
& & 
     & & \\
&  F' (J X) \times G' (K' Y) \ar[rr]^{\phi'_{JX,K'Y}} \ar[dr]_-{Fj_X \times \id}
  & & J X \times K' Y \ar[dr]^-{f_X \times \id} 
\\
& & F' (J' X) \times G' (K' Y) \ar[rr]^{\phi'_{J' X, K' Y} }
    & & J' X \times K' Y \ar[rr]^-{\psi'_{X,Y}}
        & & X \times Y 
}
\]

\subsection{Two degeneracy results}

Here are two simple degeneracy results. We first recall the notion of
operation for monads and functors.

\paragraph{A comment on operations}

The concept of (algebraic) operation of a monad can be defined
in several ways.
Given a monad $T$, an \emph{$n$-ary operation} of $T$ can be defined
to be a natural transformation $c' : (TX)^n \to T X$ (where $X^n$ is
$n$-fold product of $X$ with itself) satisfying
\[
\small
\xymatrix{
(TTX)^n \ar[d]_{(\mu_X)^n} \ar[r]^-{c'_{TX}} 
& TT X \ar[d]^{\mu_X}\\
(TX)^n \ar[r]^-{c'_X}
& T X 
} 
\]
This is the format used by Plotkin and Power \cite{PP03oper}. (We do
not require here that $T$ is strong and drop compatibility with the
strength.) Alternatively, we can say that it is a natural
transformation $c : X^n \to T X$ and drop the requirement
of commutation with $\mu$, as done by Jaskelioff and Moggi \cite{JM10}.

We can also say that it is a map $1 \to T n$ (a ``generic effect'' in
the sense of Plotkin and Power \cite{PP03}) but, for this to amount to the
same as the previous alternative, one needs that $T$ is strong.

If a finitary set monad $T$ is determined by a Lawvere theory
$(\mathcal{L}, L)$ where $\mathcal{L}$ is a category with finite
products and $L : \mathbb{F}\op \to \mathcal{L}$ is identity on
objects and strictly product-preserving, one can say that an operation
is a map $n \to 1$ in $\mathcal{L}$. Given a monad $T$ on an arbitrary
category $\C$, its large Lawvere theory is $((\Kl(T))\op, J\op)$ where
$J : \C \to \Kl(T)$ is the left adjoint of the Kleisli adjunction. A
map $n \to 1$ in $(\Kl(T))\op$ is the same as a map $1 \to T n$ in
$\C$.

In this paper, we prefer to work with operations as maps
$c : X^n \to T X$ because this format is intuitive and economic 
in proofs by diagram chasing but also because it makes 
sense when $T$ is only a functor and not a monad. 

\paragraph{Functors with a nullary operation}

For the functor $\Maybe\, X = (\just : X) + (\nothing : 1)$, it should
be clear intuitively that it cannot have a nondegenerate interacting
functor: from the element $\nothing_0$ of $\Maybe\, 0$, one cannot possibly
extract an element of $0$. Formally, we have the following theorem.
\begin{theorem}
If a functor $F$ has a nullary operation, i.e., comes with a family of
maps $c_X : 1 \to F X$ natural in $X$, then any interacting functor
$G$ is constant zero, i.e., $G Y \cong 0$ for any $Y$.
\end{theorem}
\begin{proof}
Indeed, for any $Y$, we have the map
\[
\small
\xymatrix{
G Y \ar[r]^-{\pair{!}{\id}} 
& 1 \times G Y \ar[r]^-{c_0 \times \id} 
  & F 0 \times G Y \ar[r]^-{\phi_{0,Y}} 
    & 0 \times Y \ar[r]^-{\fst} 
      & 0
}
\]
Since the initial object of an extensive category is strict (any map
to $0$ is an isomorphism), we can conclude that $G Y \cong 0$.
\qed\end{proof}

The theorem applies to $\Maybe$ since it comes with a nullary
operation $\nothing_X : 1 \to \Maybe\, X$.

\paragraph{Functors with a commutative binary operation}

A similar no-go theorem holds for commutative binary
operations.
%
\begin{theorem}
If a functor $F$ has a commutative binary operation, i.e., comes with a family of maps $c_X : X \times X \to F X$ natural in $X$ such that
$c_X = c_X \comp \sym_{X,X}$, then any interacting functor $G$ is
constant zero, i.e., $G Y \cong 0$ for any $Y$.
\end{theorem}
\begin{proof}
Let $\Bool = (\btt : 1) + (\bff : 1)$. Then, for any $Y$, the map
\[
\small
f_Y = 
\xymatrix{
G Y \ar[r]^-{\pair{!}{\id}} 
& 1 \times G Y \ar[r]^-{\pair{\btt}{\bff}} 
  & (\Bool \times \Bool) \times G Y \ar[r]^{c_\Bool \times \id} 
    & F \Bool \times G Y \ar[r]^-{\theta_{\Bool,Y}}
      & \Bool \times Y \ar[r]^{\fst}
        & \Bool
}
\]
has the property that $\bnot \comp f_Y = f_Y$:
\[
\small
\xymatrix@R=1.5pc{
& & (\Bool \times \Bool) \times G Y \ar[d]^{\sym \times \id} \ar[dr]^{c_\Bool \times \id} \\
G Y \ar[r]^-{\pair{!}{\id}} 
& 1 \times G Y \ar[ur]^-{\pair{\btt}{\bff}} 
               \ar[r]^-{\pair{\bff}{\btt}} \ar[dr]_-{\pair{\btt}{\bff}} 
  & (\Bool \times \Bool) \times G Y \ar[d]^{(\bnot \times \bnot) \times \id} \ar[r]^{c_\Bool \times \id} 
    & F \Bool \times G Y \ar[d]^{F \bnot \times \id} \ar[r]^-{\theta_{\Bool,Y}}
      & \Bool \times Y \ar[d]^{\bnot \times \id} \ar[r]^{\fst}
        & \Bool \ar[d]^{\bnot}\\
& & (\Bool \times \Bool) \times G Y \ar[r]^{c_\Bool \times \id} 
    & F \Bool \times G Y \ar[r]^-{\theta_{\Bool,Y}}
      & \Bool \times Y \ar[r]^{\fst}
        & \Bool
}
\]

By stability of coproducts under pullback in an extensive category,
we can pull the coprojections of $\Bool$ back along $f_Y$
\[
\small
\xymatrix@R=1.5pc{
P Y \ar[d]_{i_Y} \ar[r]^{h_Y}
& 1 \ar[d]^{\btt} \\
G Y \ar[r]^{f_Y}
& \Bool \\
Q Y \ar[u]^{j_Y} \ar[r]^{k_Y}
& 1 \ar[u]_{\bff}
}
\]
and the result is a pullback again.

Now we have
\[
\small
\xymatrix@R=0.8pc{
& & 1 \ar[ddr]^{\btt} \\
\\
P Y \ar[uurr]^{h_Y} \ar[ddrr]_{h_Y} \ar[r]^{i_Y}
& G Y \ar[rr]^{f_Y} \ar[dr]^{f_Y}
  & & \Bool \\
& & \Bool \ar[ur]^{\bnot} \\
& & 1 \ar[u]^{\btt} \ar[uur]_{\bff}
}
\]
so by disjointness of coproducts in an extensive category we have a
map $h'_Y : PY \to 0$ as a unique map into the pullback $0$ of $\btt$
and $\bff$:
\[
\small
\xymatrix@R=1.5pc{
& & 1 \ar[dr]^{\btt} \\
P Y \ar[urr]^{h_Y} \ar[drr]_{h_Y} \ar@{.>}[r]^{h'_Y} 
& 0 \ar[ur] \ar[dr]
  & & \Bool \\
& & 1 \ar[ur]_{\bff}
}
\]

Similarly we get a map $k'_Y : QY \to 0$. Hence we have a map
$f'_Y : G Y \to 0$ from copairing $h'_Y$ and $k'_Y$:
\[
\small
\xymatrix@R=1.5pc{
P Y \ar[d]_{i_Y} \ar[dr]^{h'_Y} \\
G Y \ar@{.>}[r]^{f'_Y}
& 0 \\
Q T \ar[u]^{j_Y} \ar[ur]_{k'_Y}
}
\]
Since the initial object is strict in an extensive category, it
follows that $G Y \cong 0$.
\qed\end{proof}

\medskip

The degeneracy problem can be overcome by switching to a residual
version of interaction laws, discussed in detail in 
Section~\ref{sec:residual} below.
As a sneak preview, given a monad $R$ on $\C$, an
$R$-\emph{residual functor-functor interaction law} is given by
two endofunctors $F$, $G$ 
and a family of maps
$\phi : F X \times G Y \to R(X \times Y)$ natural in $X$, $Y$. The
monoidal structure of the category $\Int(\C,R)$ of $R$-residual functor-functor interaction laws relies on the monad
structure of $R$. Typically, one would use the maybe, finite
nonempty multiset or finite multiset monad as $R$.

\subsection{On the structure of $\Int(\C)$}

We now look at some ways to construct functor-functor
interaction laws systematically.

\paragraph{``Stretching''}

Given a functor-functor interaction law $(F, G, \phi)$ and natural
transformations $f : F' \to F$ and $g : G' \to G$, we have a
functor-functor interaction law $(F', G', \phi \comp (f \times g))$.

\paragraph{Self-duality}

For any functor-functor interaction law $(F, G, \phi)$, we have
another functor-functor interaction law
$(F, G, \phi)\rev = (G, F, \phi\rev)$ where
$\phi\rev_{X, Y} = \sym_{Y, X} \comp \phi_{Y, X} \comp \sym_{FX,
  GY}$.
This object mapping extends to maps by $(f, g)\rev = (g, f)$, so we
have a functor $(-)\rev : (\Int(\C))\op \to \Int(\C)$.  The functor
$(-)\rev$ is an isomorphism between $(\Int(\C))\op$ and $\Int(\C)$.

\paragraph{The final functor-functor interaction law}

The final functor-functor
interaction law is $(1, 0, \phi)$ where
$ 
\phi_{X,Y} = 
\xymatrix@C=1pc{
1 \times 0 \ar[r]^-{\snd} 
& 0 \ar[r]^-{?}
  & X \times Y
}
$. 
By self-duality, the initial functor-functor interaction law is
$(0, 1, \phi\rev)$.

\paragraph{Product of two functor-functor interaction laws}
Given two functor-functor
interaction laws $(F_0, G_0, \phi_0)$ and $(F_1, G_1, \phi_1)$, their
product is $(F_0 \times F_1, G_0 + G_1, \phi)$ where
\begin{eqnarray*}
\phi_{X,Y} 
& = &  
\xymatrix{
(F_0 X \times F_1 X) \times (G_0 Y + G_1 Y) \ar[r]^-{\rdist}  
&
} \\
& & \quad
\xymatrix@C=5pc{
(F_0 X \times F_1 X) \times G_0 Y + (F_0 X \times F_1 X) \times G_1 Y \ar[r]^-{\fst \times \id + \snd \times \id} 
& 
} \\
& & \quad 
\xymatrix@C=4pc{
  F_0 X \times G_0 Y + F_1 X \times G_1 Y \ar[r]^-{{\phi_0}_{X,Y} + {\phi_1}_{X,Y}}
&  X \times Y + X \times Y \ar[r]^-{\nabla} 
  & X \times Y
}   
\end{eqnarray*}
By self-duality, the \emph{coproduct} of $(G_0, F_0, \phi_0\rev)$ and $(G_1, F_1, \phi_1\rev)$ is $(G_0 + G_1, F_0 \times F_1, \phi\rev)$. 

\paragraph{An initial algebra-final coalgebra construction}

Assume that $\C$ has the relevant initial algebras and final
coalgebras. Given functors $F, G : \C \times \C \to \C$
and a family of maps
$\phi_{X,Y,W,Z} : F (X, Z) \times G (Y, W) \to X \times Y + Z \times
W$ natural in $X, Y, Z, W$.
Then we have an interaction law $(F', G', \phi')$ where $F' X = \mu
Z.\, F(X, Z)$, $G' X = \nu W.\, G(Y, W)$ and $\phi'$ is constructed
as follows. We equip $G' Y \fun (X \times Y)$ with an $F(X, {-})$-algebra structure $\theta^0_{X,Y}$ by currying the map 
\begin{eqnarray*}
\theta_{X,Y} 
& = &  
\xymatrix@C=4pc{
F (X, G' Y \fun (X \times Y)) \times G' Y \ar[r]^-{\id \times \out_{G (Y, {-})}}
& }
\\
& & \quad
\xymatrix@C=6pc{
F (X, G' Y \fun (X \times Y)) \times G (Y, G' Y) 
  \ar[r]^-{\phi_{X,Y,G' Y \fun (X \times Y), G' Y}}  
 &
}
\\
& & \quad 
\xymatrix@C=2pc{
X \times Y + (G' Y \fun (X \times Y)) \times G' Y
   \ar[r]^-{\id + \ev} 
& X \times Y + X \times Y 
   \ar[r]^-{\nabla} 
  & X + Y 
}
\end{eqnarray*}
The map $\phi'_{X, Y}$ is obtained by uncurrying the corresponding unique map
$\phi^0_{X,Y} : F' X \to G' Y \fun (X \times Y)$ from the initial
$F(X, {-})$-algebra.

\paragraph{Restricting to fixed $F$ or $G$} 

Sometimes it is of interest to focus on interaction laws of a fixed
first functor $F$ or a fixed second functor $G$ (and accordingly 
on interaction law maps with the first resp.\ the second natural 
transformation the identity natural transformation on $F$ resp.\ $G$). 
We denote the
corresponding categories by $\Int(\C)|_{F,-}$ and
$\Int(\C)|_{-,G}$. The isomorphism of categories
$\Int(\C)\op \cong \Int(\C)$ given by $(-)\rev$ restricts to
$(\Int(\C)|_{F,-})\op \cong \Int(\C)|_{-,F}$.

The final object of $\Int(\C)|_{F,-}$ is $(F, 0, \phi)$ where
\[
\phi_{X,Y} = 
\xymatrix{
F X \times 0 \ar[r]^-{\snd} 
& 0 \ar[r]^-{?} 
  & X \times Y
}
\]
By self-duality, the initial object of $\Int(\C)|_{-,F}$ is 
$(0, F, \phi\rev)$.
About the initial object of $\Int(\C)|_{F,-}$ we will see in the next
subsection.

\subsection{Functor-functor interaction in terms of the dual}

If $\C$ is Cartesian closed, then we can define the
\emph{dual} $\d{G}$ of an \emph{endofunctor} $G$ on $\C$ by
\[
\d{G} X = \mbox{$\int_Y$} G Y \fun (X \times Y) 
\]
provided that this end exists, 
and the \emph{dual} $\d{g} : \d{G} \to \d{G'}$ of a \emph{natural
  transformation} $g : G' \to G$ by
\[
\d{g}_X = \mbox{$\int_Y$} g_Y \fun (X \times Y) 
\]
This construction is contravariantly functorial, i.e., if the dual is
everywhere defined, then we have $\d{(-)} : [\C, \C]\op \to [\C, \C]$.
The existence of all the ends required for this is a strong condition
(e.g., a small category that has all limits under classical logic is
necessarily a preorder by an argument by Freyd~\cite{MLan:CWM}). But for
well-definedness and functoriality of $\d{(-)}$ in the general case,
it suffices to restrict it to those endofunctors on $\C$ that happen
to have the dual or, if one so wishes, to some well-delineated smaller
class of functors that are guaranteed to have it (e.g., to
finitary functors if $\C$ is locally finitely presentable).  For
$\d{(-)}$ to be a contravariant endofunctor on some full subcategory
of $[\C, \C]$, we can restrict it to those endofunctors on $\C$ that
are dualizable any finite number of times or to some other class
of functors closed under the dual. Throughout this paper, we
deliberately ignore this existence issue: we either explicitly prove
for the ends of interest that they exist or we use such ends on the
assumption that they happen to exist. 

We have
\[
\renewcommand{\arraystretch}{1.5}
\begin{array}{l}
\int_Y \C(G Y, \overbrace{\mbox{$\int_X$} F X \fun (Y \times X)}^{\d{F} Y})
\cong \int_{Y,X} \C(G Y \times F X, Y \times X) \\
\quad \cong \int_{X,Y} \C(F X \times G Y, X \times Y)
\cong \int_X \C(F X, \underbrace{\mbox{$\int_Y$} G Y \fun (X \times Y)}_{\d{G} X})
\end{array}
\]
where by the top-level ends we just indicate collections (not
necessarily sets) of natural transformations, so existence is not an
issue. 


Thus, to have a functor-functor interaction law of $F$, $G$ is the
same as to have a natural transformation $\phia : F \to \d{G}$ or a
natural transformation $\phib : G \to \d{F}$.

Under the first of these identifications, an interaction
law map between $(F, G, \phia)$ and $(F', G', \phia')$ is given by
natural transformations $f : F \to F'$ and $g : G' \to G$ satisfying
$\d{g} \comp \phia = \phia' \comp f$.
Under the second one, an interaction law map between $(F, G, \phib)$
and $(F', G', \phib')$ is given by natural transformations
$f : F \to F'$ and $g : G' \to G$ satisfying
$\phib \comp g = \d{f} \comp \phib'$.

We have thus established that these categories are isomorphic:
\begin{enumerate}[label=(\roman*)]
\item[(o)] the category $\Int(\C)$ of functor-functor interaction laws;
\item the comma category $[\C,\C] \dn \d{(-)}$ of triples of two
    functors $F, G$ and a natural transformation $F \to \d{G}$;
\item the comma category ${\d{(-)}}\op \dn [\C,\C]\op$ of triples of
  two functors $F, G$ and a natural transformation $G \to \d{F}$.
\end{enumerate}

From these observations it is immediate that
$\Int(\C)|_{-,G} \cong [\C, \C] / \d{G}$ and
$\Int(\C)|_{F,-} \cong \d{F} \backslash [\C,\C]\op$.
Hence, the initial
object of $\Int(\C)|_{-,G}$ is $(0, G, \ldots)$ while the final object
is $(\d{G}, G, \ldots)$. The initial object of $\Int(\C)|_{F,-}$ is
$(F, \d{F}, \ldots)$ while the final object is $(F, 0, \ldots)$.

\subsection{Dual for some constructions on functors}

Here are constructions of the dual for some basic constructions of
functors.


\paragraph{Dual of the identity functor}
$\d{\Id} \cong \Id$.
\begin{proof}
  Let $G Y = Y$. Then
\begin{eqnarray*}
\d{G} X & = & \int_Y Y \fun (X \times Y) \\ 
& \cong & \int_Y (1 \fun Y) \fun (X \times Y) \\
& \cong & X \times 1 \\
& \cong & X
\end{eqnarray*}
\end{proof}

\paragraph{Duals of terminal functor, products of a functor, 
  initial functor, coproduct of two functors}
\begin{itemize}

\item  Let $G\, Y = 1$. Then $\d{G}\, X \cong 0$.
  \begin{proof}
    \begin{eqnarray*}
\d{G} X & = & \int_Y 1 \fun (X \times Y) \\
& \cong & \int_Y X \times Y \\
& \cong & X \times \int_Y Y \\ 
& \cong & X \times 0 \\ 
& \cong & 0
\end{eqnarray*}

  \end{proof}

\item  Let $G\, Y = A \times G' Y$. Then $\d{G}\, X \cong A \fun \d{G'} X$.
  \begin{proof}
    \begin{eqnarray*}
\d{G} X & = & \int_Y A \times G_0 Y \fun (X \times Y) \\
& \cong &  \int_Y A \fun (G_0 Y \fun (X \times Y)) \\
& \cong & A \fun \int_Y G_0 Y \fun (X \times Y) \\
& = & A \fun \d{G_0} X
\end{eqnarray*}
  \end{proof}

\item  A little more generally, for $G\, Y = \sum a : A .\, G' a\, Y$, one
  has $\d{G}\, X \cong \prod a : A.\, \d{(G' a)}\, X$.
  
\item  Specializing to $A = 0$ resp.\ $A = \Bool$, we learn: Let $G\, Y =
  0$.
  Then $\d{G} X \cong 1$.  Let $G\, Y = G_0\, Y + G_1\, Y$. Then
  $\d{G}\, X \cong \d{G_0}\, X \times \d{G_1}\, X$.
\end{itemize}

\paragraph{Dual of exponents of the identity functor}
Let $G Y = A \fun Y$. Then $\d{G} X \cong A \times X$.
\begin{proof}
  \begin{eqnarray*}
\d{G} X & = & \int_Y (A \fun Y) \fun (X \times Y) \\
& \cong & X \times A \\
& \cong & A \times X
\end{eqnarray*}

\end{proof}

\begin{example} \label{ex:nelists-a}
  Let
  $G\, Y = Y^+ = \mu Z.\, Y \times (1 + Z) \cong \sum n : \Nat.\,
  ([0..n] \fun Y)$ (nonempty lists). We
  have $\d{G} X \cong \prod n : \Nat.\, ([0..n] \times X)$.
\end{example}

Sometimes only a ``lower bound'' on the dual of a functor constructed
from some given functors can be expressed in terms of their
duals. This holds for the composition of two general functors, 
incl.\ for exponents of a general functor.

\paragraph{Dual of exponents of a general functor}
  Let $G\, Y = A \fun G'\, Y$. For a general $G'$, we only have a canonical natural
  transformation with components
  $A \times \d{G'}\, Y \to \d{G}\, Y$.
  \begin{proof}
\begin{eqnarray*}
\d{G} X & = & \int_Y (A \fun G' Y) \fun (X \times Y) \\
& \leftarrow & \int_Y A \times (G' Y \fun (X \times Y)) \\
& \cong & A \times \int_Y G' Y \fun (X \times Y) \\
& = & A \times \d{G'} X
\end{eqnarray*}
  \end{proof}

\paragraph{Dual of composition of two general functors}
For general $G_0$, $G_1$, we only have the canonical natural
transformation
$\mm^{G_0,G_1} : \d{G_0} \cdot \d{G_1} \to \d{(G_0 \cdot G_1)}$.
\begin{proof}
  \begin{eqnarray*}
\d{G} X & = & \int_Y G_0 (G_1 Y) \fun (X \times Y) \\
& \leftarrow & \int_Y \int_Z (Z \fun G_1 Y) \fun (G_0 Z \fun (X \times Y)) \\ 
& \cong & \int_Z G_0 Z \fun \int_Y (Z \fun G_1 Y) \fun (X \times Y) \\ 
& \leftarrow & \int_Z G_0 Z \fun \int_Y (G_1 Y \fun (X \times Y)) \times Z \\
& \cong & \int_Z G_0 Z \fun ((\int_Y G_1 Y \fun (X \times Y)) \times (\int_Y Z)) \\
& = & \int_Z G_0 Z \fun (\d{G_1} X \times Z) \\
& = & \d{G_0} (\d{G_1} X)
\end{eqnarray*}

\end{proof}

This
hints that $\d{(-)} : [\C,\C]\op \to [\C,\C]$ is not monoidal, but
only lax monoidal (see~Section~\ref{sec:monad-comonad-d-sd}).

\begin{example} \label{ex:update-a}
  Let $G_0\, Y = A \fun Y$, $G_1\, Y = B \times Y$, so
  $G\, Y = (G_0 \cdot G_1)\, Y = A \fun (B \times Y) \cong \linebreak (A \fun B) \times (A
  \fun Y)$.
  The dual of $G$ is $\d{G}\, X \cong (A \fun B) \fun (A \times X)$
  rather than $(\d{G_0} \cdot \d{G_1})\, X \cong A \times (B \fun X)$ 
  as we might perhaps expect.  
  We
  saw the interaction law of $G$ with $\d{G}$ in Example~\ref{ex:update-o}. The
  canonical natural transformation
  $\mm^{G_0,G_1} : \d{G_0} \cdot \d{G_1} \to \d{G}$ is
  $\mm^{G_0,G_1}\, (a, f) = \lambda g.\, (a, f\, (g\, a))$.\er{What is this reference? The general interaction between a functor and its dual?} \tu{fixed now, it
was one of the very first examples in the functor-functor part}
\end{example}

\section{Monad-comonad interaction}
\label{sec:monad-comonad}

\subsection{Monad-comonad interaction laws}


In a monad-comonad interaction law, the allowed computations (the chosen notion of computation) must
include ``just returning'' and be closed under sequential composition,
so they are defined by a monad rather than a functor. To match this,
the allowed machine behaviors (the notion of machine behavior) are defined by a comonad.  The idea is
that interaction of a ``just returning'' computation should terminate
immediately (in the initial state of the given machine behavior)
whereas interaction of a sequence of computations should amount to a
sequence of interactions. 

\medskip

We define a \emph{monad-comonad interaction law} on $\C$ to be given by a monad
$T = (T, \eta, \mu)$ and a comonad $D = (D, \eps, \de)$ together with
a family $\psi$ of maps
\[
\psi_{X,Y} : T X \times D Y \to X \times Y
\]
natural in $X$ and $Y$ (i.e., a functor-functor interaction law of
$T$, $D$ where $T$ and $D$ carry a monad resp.\ comonad structure)
such that also
\begin{equation} \label{eq:mcil-conds}
\small
\xymatrix@R=0.8pc@C=1.5pc{
& X \times Y \ar@{=}[r]
  & X \times Y \ar@{=}[dd] \\
X \times D Y \ar[ur]^{\id \times \eps_Y}  \ar[dr]_{\eta_X \times \id} 
& & \\
& T X \times D Y \ar[r]^-{\psi_{X,Y}}
  & X \times Y 
}
\quad
\xymatrix@R=0.8pc@C=1.5pc{
& T T X \times D D Y \ar[r]^{\psi_{TX, DY}}
 & T X \times D Y \ar[r]^{\psi_{X,Y}}
   & X \times Y \ar@{=}[dd] \\
T T X \times D Y \ar[ur]^{\id \times \de_Y} \ar[dr]_{\mu_X \times \id} 
& & & \\
& T X \times D Y \ar[rr]^-{\psi_{X,Y}}
  & & X \times Y
}
\end{equation}

\begin{example}
  Take $T X = A \fun X$, $D Y = A \times Y$ and
  $\psi\, (f, (a, y)) = (f\, a, y)$ for a fixed object $A$.  The
  functors $T$ and $D$ are a monad (a reader monad) resp.\ a comonad
  and $\psi$ meets the conditions (\ref{eq:mcil-conds}).
\end{example}

\begin{example}
  Take $T X = B \times X$, $D Y = B \fun Y$,
  $\psi\, ((b, x), g) = (x, g\, b)$ for a fixed monoid $B$. The
  functors $T$, $D$ are a monad (a writer monad) resp.\ a comonad
  and $\psi$ meets the requisite conditions.
\end{example}

\begin{example} \label{ex:update-b}
  Take $T X = A \fun (B \times X)$, $D Y = A \times (B \fun Y)$,
  $\psi\, (f, (a, g)) = \letin{(b, x) \leftarrow f\, a}{(x, g\, b)}$
  for a fixed monoid $B$ acting on a fixed object $A$. The functors
  $T$, $D$ are a monad (an update monad \cite{AU14updmnd}) resp.\ a comonad and $\psi$
  meets the requisite conditions.
\end{example}


  Monad-comonad interaction laws are essentially the same as monoid objects in
  the monoidal category $\Int(\C)$ of functor-functor interaction laws.
  To be precise, a monad-comonad interaction law
  $((T, \eta, \mu), \linebreak (D, \eps, \de), \psi)$ yields a monoid
  $((T, D, \psi), (\eta, \eps), (\mu, \de))$ and vice versa.


  A \emph{monad-comonad interaction law map} between $(T, D, \psi)$,
  $(T', D', \psi')$ is a pair $(f : T \to T', g : D' \to D)$ of a
  monad map and a comonad map that, as a pair of natural
  transformations between the underlying functors, is a
  functor-functor interaction law morphism between the underlying
  functor-functor interaction laws.

Monad-comonad interaction law maps correspond to monoid
morphisms in $\Int(\C)$. Thus monad-comonad interaction laws form a
category $\mcInt(\C)$ isomorphic to
the category $\Mon(\Int(\C))$.

\subsection{A degeneracy result}

\paragraph{Monads with an associative operation}

Here is a degeneracy theorem for monad-comonad interaction
laws.
\begin{theorem}
If a monad $T$ has an associative binary operation, i.e., family of
maps $c_X : X \times X \to T X$ natural in $X$ satisfying
\[
\small
\xymatrix@R=0.8pc{
(X \times X) \times X \ar[dd]_{\ass} \ar[r]^-{c_X \times \eta_X}
& T X \times T X \ar[r]^-{c_{TX}} 
  & T T X \ar[dr]^{\mu_X} \\
& & & T X \\
X \times (X \times X) \ar[r]^-{\eta_X \times c_X}
& T X \times T X \ar[r]^-{c_{TX}} 
  & T T X \ar[ur]_{\mu_X}    
}
\]
then, for any comonad $D$ and interaction law $\psi_{X,Y} : T X \times
D Y \to X \times Y$, we have
\[
\small
\xymatrix@R=1.5pc@C=2pc{
(X \times X) \times X \ar[d]_{\fst \times \id \times \id} \times D Y \ar[r]^-{c_X \times \eta_X \times \id}
& T X \times T X \times D Y \ar[r]^-{c_{TX} \times \id} 
  & T T X \times D Y \ar[r]^-{\mu_X \times \id} 
    & T X \times D Y \ar[dr]^-{\psi_{X,Y}} \\
X \times X \times D Y \ar[rrr]^-{c_X \times \id}
& & 
    & T X \times D Y \ar[r]^-{\psi_{X,Y}} 
      & X \times Y \\
X \times (X \times X) \times D Y \ar[u]^{\id \times \snd \times \id} \ar[r]^-{\eta_X \times c_X \times \id}
& T X \times T X \times D Y  \ar[r]^-{c_{TX} \times \id} 
  & T T X \times D Y  \ar[r]^-{\mu_X \times \id}  
    &  T X \times D Y \ar[ur]_-{\psi_{X,Y}} 
}
\]
\end{theorem}
\begin{proof}
  For any $Y$, by distributivity in an extensive category,
$\Bool \times Y$ is a coproduct of $Y$ and $Y$ with coprojections
$\pair{\btt \comp {!}}{\id}$ and $\pair{\bff \comp {!}}{\id}$.

By stability of coproducts under pullback in an extensive category, we
can pull
$\theta_{\Bool, Y} \comp \pair{c_\Bool \comp \pair{\btt}{\bff} \comp
  {!}}{\id} : D Y \to \Bool \times Y$
back along the coprojections of $\Bool \times Y$ and get
that $D Y$ is a coproduct of two objects $P Y$ and $Q Y$ with coprojections 
$i_Y$ and $j_Y$:
\[
\small
\xymatrix@R=1.5pc@C=3pc{ 
P Y \ar[d]_{i_Y} \ar[rrr]^{h_Y}
& & & Y \ar[d]^{\pair{\btt \comp {!}}{\id}} \\
D Y \ar[r]^-{\pair{\pair{\btt}{\bff} \comp
  {!}}{\id}}
& \Bool \times \Bool \times D Y \ar[r]^-{c_\Bool \times \id}
  & T \Bool \times D Y \ar[r]^-{\psi_{\Bool,Y}}
    & \Bool \times Y \\
Q Y \ar[u]^{j_Y} \ar[rrr]^{k_Y}
& & & Y \ar[u]_{\pair{\bff \comp {!}}{\id}} \\
}
\]

It is easily checked that we have 
\[
\small
\xymatrix@R=1.5pc{
X \times X \times P Y \ar[d]_{\id \times i_Y} \ar[r]^-{\fst \times \id} 
& X \times P Y \ar[rd]^-{\id \times h_Y} \\
X \times X \times D Y \ar[r]^-{c_X \times \id}
& T X \times D Y \ar[r]^-{\psi_{X,Y}} 
  & X \times Y \\
X \times X \times Q Y \ar[u]^{\id \times j_Y} \ar[r]^-{\snd \times \id}
  & X \times Q Y \ar[ru]_-{\id \times k_Y}
}
\]

Also by stability of coproducts under pullback, we can pull $\de_Y : D Y \to D D Y$ back along
the coprojections of $D D Y$ and get that $D Y$ is a coproduct of two objects $P' Y$ and $Q' Y$ with coprojections $i'_Y$ and $j'_Y$:
\[
\small
\xymatrix@R=1.5pc{
P' Y \ar[d]_{i'_Y} \ar[r]^-{f_Y} 
& P D Y \ar[d]^{i_{D Y}} \\
D Y \ar[r]^{\de_Y}
& D D Y \\
Q' Y \ar[u]^{j'_Y} \ar[r]^-{g_Y}
& Q D Y \ar[u]_{j_{D Y}}
}
\]

Hence, for any $X$, by distributivity, also
$X \times (X \times X) \times D Y$ is a coproduct of
$X \times (X \times X) \times P' Y$ and
$X \times (X \times X) \times Q' Y$ with coprojections
$\id \times i'_Y$ and $\id \times j'_Y$.

Now, the two maps
$\psi_{X, Y} \comp \psi_{TX, DY} \comp (c_{T X} \comp \eta_X \times
c_X) \times \de_Y$
and
$\psi_{X, Y} \comp \psi_{TX, DY} \comp (c_{T X} \comp \eta_X \times
(\eta_X \comp \snd)) \times \de_Y$
both satisfy both triangles of the unique copair of
$\psi_{X, Y} \comp (\eta_X \comp \fst) \times (h_{D Y} \comp f_Y)$
and
$\psi_{X, Y} \comp (c_X \comp \snd) \times (k_{D Y} \comp g_Y)$,
i.e., they are the same map. Indeed, we have
\[
\scriptsize
\hspace*{-4mm}
\xymatrix@R=1.5pc{
& & X \times P D Y \ar[dr]^-{\eta_X \times \id} \\
X \times (X \times X) \times P' Y 
    \ar[d]^{\id \times i'_Y} \ar[r]^-{\id \times f_Y}
& X \times (X \times X) \times P D Y 
    \ar[d]^{\id \times i_{D Y}} \ar[r]^-{\eta_X \times c_X \times \id}
           \ar[ur]^-{\fst \times \id}
  & T X \times TX \times P D Y 
   \ar[d]^{\id \times i_{D Y}}  \ar[r]^-{\fst \times \id} 
    & T X \times P D Y \ar[dr]^-{\id \times h_{D Y}} \\
X \times (X \times X) \times D Y \ar[r]^-{\id \times \de_Y}
& X \times (X \times X) \times D D Y \ar[r]^-{\eta_X \times c_X \times \id}
  & T X \times T X \times D D Y \ar[r]^-{c_{TX} \times \id}
    & T T X \times D D Y \ar[r]^-{\psi_{TX,DY}}
      & T X \times D Y \ar[r]^-{\psi_{X,Y}}
        & X \times Y \\
X \times (X \times X) \times Q' Y 
    \ar[u]_{\id \times j'_Y} \ar[r]^-{\id \times g_Y}
& X \times (X \times X) \times Q D Y 
    \ar[u]_{\id \times j_{D Y}} \ar[r]^-{\eta_X \times c_X \times \id}
           \ar[dr]_-{\snd \times \id}
  & T X \times TX \times Q D Y 
   \ar[u]_{\id \times j_{D Y}}  \ar[r]^-{\snd \times \id} 
    & T X \times Q D Y \ar[ur]_-{\id \times k_{D Y}} \\
& & X \times X \times Q D Y \ar[ur]_-{c_X \times \id} 
}
\]
And, using associativity, we also have
\[
\scriptsize
\xymatrix@R=0.8pc{
& & X \times P D Y \ar[dr]^-{\eta_X \times \id} \\
X \times (X \times X) \times P' Y 
    \ar[d]^{\id \times i'_Y} \ar[r]^-{\id \times \snd \times f_Y}
& X \times X \times P D Y 
    \ar[d]^{\id \times i_{D Y}} \ar[r]^-{\eta_X \times \eta_X \times \id}
           \ar[ur]^-{\fst \times \id}
  & T X \times TX \times P D Y 
   \ar[d]^{\id \times i_{D Y}}  \ar[r]^-{\fst \times \id} 
    & T X \times P D Y \ar[dr]^-{\id \times h_{D Y}} \\
X \times (X \times X) \times D Y \ar[r]^-{\id \times \snd \times \de_Y}
& X \times X \times D D Y \ar[r]^-{\eta_X \times \eta_X \times \id}
  & T X \times T X \times D D Y \ar[r]^-{c_{TX} \times \id}
    & T T X \times D D Y \ar[r]^-{\psi_{TX,DY}}
      & T X \times D Y \ar[r]^-{\psi_{X,Y}}
        & X \times Y \ar@{=}[dddddd] \\
X \times (X \times X) \times Q' Y \ar[dd]^{\ass^{-1} \times \id}
    \ar[u]_{\id \times j'_Y} \ar[r]^-{\id \times \snd \times g_Y}
& X \times X \times Q D Y 
    \ar[u]_{\id \times j_{D Y}} \ar[r]^-{\eta_X \times \eta_X \times \id}
           \ar[dr]^-{\snd \times \id}
  & T X \times TX \times Q D Y 
   \ar[u]_{\id \times j_{D Y}}  \ar[r]^-{\snd \times \id} 
    & T X \times Q D Y \ar[ur]_-{\id \times k_{D Y}} \\
& & X \times Q D Y \ar[ur]^-{\eta_X \times \id} \\
(X \times X) \times X \times Q' Y \ar[r]^-{\id \times g_Y}
    \ar[d]^{\id \times j'_Y}
& (X \times X) \times X \times Q D Y \ar[uu]_{\fst \times \id}
     \ar[ur]^-{\snd \times \id} \ar[r]^-{c_X \times \eta_X \times \id} 
     \ar[d]^{\id \times j_{DY}}
  & T X \times T X \times Q D Y \ar[uur]_-{\snd \times \id} 
     \ar[d]^{\id  \times j_{DY}}       
   \\
(X \times X) \times X \times D Y \ar@{=}[d] \ar[r]^-{\id \times \de_Y}
& (X \times X) \times X \times D D Y \ar[r]^-{c_X \times \eta_X \times \id}
  & T X \times T X \times D D Y \ar[r]^-{c_X \times \id}
    & T T X \times D D Y \ar[uuuru]_-{\psi_{TX,DY}} \\
(X \times X) \times X \times D Y \ar[r]^-{c_X \times \eta_X \times \id}
& T X \times T X \times D Y \ar[r]^-{c_X \times \id}
  & T T X \times D Y \ar[ur]_-{TTX \times \de_Y} \ar[dr]^-{\mu_X \times \id}\\
& & & T X \times D Y \ar[rr]^-{\psi_{X,Y}}
      & & X \times Y \ar@{=}[dd] \\
X \times (X \times X) \times D Y \ar@{=}[d] \ar[uu]_{\ass^{-1} \times \id} \ar[r]^-{\eta_X \times c_X \times \id}
& T X \times T X \times D Y \ar[r]^-{c_X \times \id}
  & T T X \times D Y \ar[dr]^-{\id \times \de_Y} \ar[ur]_-{\mu_X \times \id}\\
X \times (X \times X) \times D Y \ar[r]^-{\id \times \de_Y}
& X \times (X \times X) \times D D Y \ar[r]^-{\eta_X \times c_X \times \id}
  & T X \times T X \times D D Y \ar[r]^-{c_X \times \id}
    & T T X \times D D Y \ar[r]^-{\psi_{TX,DY}}
      & T X \times D Y \ar[r]^-{\psi_{X,Y}}
        & X \times Y \\
X \times (X \times X) \times Q' Y 
  \ar@/^5pc/@{=}[uuuuuuuu]
  \ar[u]_-{\id \times j'_Y} 
  \ar[r]^-{\id \times g_Y}
& X \times (X \times X) \times Q D Y 
     \ar[u]_-{\id \times j_{DY}} 
     \ar[dr]_-{\snd \times \id} \ar[r]^-{\eta_X \times c_X \times \id}
  & T X \times T X \times Q D Y 
     \ar[u]_-{\id \times j_{DY}} \ar[r]^-{\snd \times \id}
    & T X \times Q D Y \ar[ur]_-{\id \times k_Y} \\
& & X \times X \times Q D Y \ar[ur]_-{c_X \times \id}
%
%
}
\]

The desired result now follows by the following calculation:
\[
\scriptsize
\xymatrix@R=1pc{
& T X \times D Y \ar[dr]^-{T\eta_X \times \id} \ar@{=}[rr]
  & & T X \times D Y \ar[rr]^-{\psi_{X,Y}} 
      & & X \times Y \ar@{=}[d] \\
X \times X \times D Y \ar[ur]^{c_X \times \id} \ar[r]^-{\eta_X \times \eta_X \times \id}
& T X \times T X \times D Y  \ar[r]^-{c_{TX} \times \id} 
  & T T X \times D Y  \ar[ur]^-{\mu_X \times \id} \ar[r]^{\id \times \de_Y}
    & T T X \times D D Y \ar[r]^-{\psi_{TX,DY}}
      & T X \times D Y \ar[r]^-{\psi_{X,Y}} 
        & X \times Y \ar@{=}[d] \\
X \times (X \times X) \times D Y \ar[u]^{\id \times \snd \times \id} \ar[r]^-{\eta_X \times c_X \times \id}
& T X \times T X \times D Y  \ar[r]^-{c_{TX} \times \id} 
  & T T X \times D Y  \ar[dr]_-{\mu_X \times \id} \ar[r]^{\id \times \de_Y}
    & T T X \times D D Y \ar[r]^-{\psi_{TX, DY}}
      & T X \times D Y \ar[r]^-{\psi_{X, Y}}
        & X \times Y \ar@{=}[d] \\
& & & T X \times D Y \ar[rr]_-{\psi_{X,Y}} 
      & & X \times Y
}
\]
\qed\end{proof}

\begin{example} \label{ex:nelists-b}
The monad $T X = X^+$ of nonempty lists (the free semigroup
delivering monad) comes with an associative operation
$\dblt_X : X \times X \to T\, X$ defined by
$\dblt\, (x_0, x_1) = [x_0, x_1]$. The degeneracy theorem tells us
that, while functor-functor interaction laws can accomplish this, no
monad-comonad interaction law can extract $x_1$ from a list
$[x_0, x_1, x_2]$ and more generally any middle element $x_i$
($0 < i < n+1$) from a list $[x_0, \ldots, x_{n+1}]$.
\end{example}

Just as functor-functor interaction laws can be generalized to a
residual variant to counteract degeneracies, so can monad-comonad
interaction laws (see Section~\ref{sec:residual}).

\subsection{On the structure of $\mcInt(\C)$}

We now explore the structure of the category $\mcInt(\C)$. As this is
the category of monoid objects of $\Int(\C)$, the structure of $\mcInt(\C)$
is in many respects similar to $\Int(\C)$. But there are also
important differences. 

\paragraph{``Stretching''}
Given a monad-comonad interaction law $(T, D, \psi)$, a monad morphism
$f : T' \to T$ and a comonad morphism $g : D' \to D$, we
have a monad-comonad interaction law $(T', D', \psi \comp f \times g)$.

\paragraph{Final and initial monad-comonad interaction laws}

The final monad-comonad interaction law is $(1, 0, \psi)$ where
$\psi_{X, Y} : 1 \times 0 \to X \times Y$ is the evident map. 

The initial monad-comonad interaction law is
$(\Id, \Id, \id_{\Id \times \Id})$.

\paragraph{Product of two monad-comonad interaction laws}

Given two monad-comonad
interaction laws $(T_0, D_0, \psi_0)$ and $(T_1, D_1, \psi_1)$, their
product is $(T_0 \times T_1, D_0 + D_1, \psi)$ where $\psi_{X,Y} :
(T_0 X \times T_1 X) \times (D_0 Y + D_1 Y) \to X \times Y$ is defined
as in Section~\ref{sec:functor-functor}. The product of the underlying
functors of the two monads is the underlying functor of their
product.

\paragraph{Coproduct of two monad-comonad interaction laws}

The coproduct of two monad-comonad interaction laws is given by the
coproduct of the two monads, the product of the two comonads and a
suitable natural transformation. The coproduct of two monads is
complicated to construct. For two ideal monads, it can be expressed in
terms of initial algebras of endofunctors on $\C \times \C$ (mutually
inductive types) \cite{GU04}.

\paragraph{Interaction laws of a composite monad}

Given two monad-comonad interaction laws $(T_0, D, \psi_0)$ and
$(T_1, D, \psi_1)$ and a monad-monad distributive law $\lambda$ of
$T_1$ over $T_0$.  Then $T_0 \cdot T_1$ is a monad. If $\psi_0$ and
$\psi_1$ are matching in the sense of commutation of
\[
\small
\xymatrix@R=0.8pc@C=3pc{
&
T_1T_0X \times DDY 
  \ar[r]^-{{\psi_1}_{T_0X, DY}} 
  & T_0X \times DY 
  \ar[r]^-{{\psi_0}_{X,Y}} 
    & X \times Y 
  \ar@{=}[dd] \\
T_1T_0 X \times D Y \ar[ur]^-{\id \times \de_Y} \ar[dr]_-{\lambda_X \times \de_Y}
\\
& 
T_0T_1X \times DDY
  \ar[r]^-{{\psi_0}_{T_1X, DY}} 
  & T_1X \times DY 
  \ar[r]^-{{\psi_1}_{X,Y}} 
    & X \times Y 
}
\]
then we have a monad-comonad interaction law $(T_0 \cdot T_1, D,
\psi)$ where
\[
\small
\psi_{X,Y} = 
\xymatrix@C=3pc{
T_0T_1X \times DY \ar[r]^{\id \times \de_Y}
& T_0T_1X \times DDY
  \ar[r]^-{{\psi_0}_{T_1X, DY}}
 & T_1X \times DY 
  \ar[r]^-{{\psi_1}_{X,Y}}
   & X \times Y
}
\]

The condition above is precisely the condition for $(\lambda, \id_D)$
to be a map between the functor-functor interaction laws $(T_1 \cdot
T_0, D, \psi_0 \comp \psi_1 \cdot (T_0 \times D) \comp (\id_{T_0 \cdot T_1} \times \de))$ and $(T_0 \cdot
T_1, D, \psi_1 \comp \psi_0 \cdot (T_1 \times D) \comp (\id_{T_0 \cdot T_1} \times \de))$.

\paragraph{Interaction laws of a composite monad and a composite
  comonad}

Given two monad-comonad interaction laws $(T_0, D_0, \psi_0)$ and
$(T_1, D_1, \psi_1)$, a monad-monad distributive law $\lambda$ of $T_1$
over $T_0$ and a comonad-comonad distributive law $\kappa$ of $D_0$
over $D_1$.  Then $T_0 \cdot T_1$ is a monad and $D_0 \cdot D_1$ is a
comonad. If $\psi_0$ and $\psi_1$ are matching in the sense of
commutation of
\[
\small
\xymatrix@R=0.8pc@C=3pc{
& T_1T_0X \times D_1D_0Y 
  \ar[r]^-{{\psi_1}_{T_0X, D_0Y}} 
  & T_0X \times D_0Y 
  \ar[r]^-{{\psi_0}_{X,Y}} 
    & X \times Y 
  \ar@{=}[dd] \\
T_1T_0 X \times D_0D_1 Y \ar[ur]^-{\id \times \kappa_Y} \ar[dr]_-{\lambda_X \times \id}
\\
& 
T_0T_1X \times D_0D_1Y
  \ar[r]^-{{\psi_0}_{T_1X, D_1Y}} 
  & T_1X \times D_1Y 
  \ar[r]^-{{\psi_1}_{X,Y}} 
    & X \times Y 
}
\]
then we have a monad-comonad interaction law $(T_0 \cdot T_1, D_0 \cdot D_1,
\psi)$ where
\[
\small
\psi_{X,Y} = 
\xymatrix@C=3pc{
T_0T_1X \times D_0D_1Y
  \ar[r]^-{{\psi_0}_{T_1X, D_1Y}} &
T_1X \times DY 
  \ar[r]^-{{\psi_1}_{X,Y}} &
X \times Y
}
\]

The condition above is precisely the condition for $(\lambda, \kappa)$
to be a map between the functor-functor interaction laws $(T_1 \cdot
T_0, D_1 \cdot D_0, \psi_0 \comp \psi_1 \cdot (T_0 \times D_0))$ and $(T_0 \cdot
T_1, D_0 \cdot D_1, \psi_1 \comp \psi_0 \cdot (T_1 \times D_1))$.

\paragraph{An initial algebra-final coalgebra construction}

The initial algebra-final coalgebra construction from
Section~\ref{sec:functor-functor} gives a monad-comonad interaction law
if we start with a parameterized monad $T$, a parameterized comonad
$D$ \cite{Uus03} and a family of maps $\psi_{X,Y,Z,W} : T(X, Z) \times D(Y, W) \to
X \times Y + Z \times W$ natural in $X, Y, Z, W$ that agree in the
sense of commutation of the diagrams
\[
\small
\xymatrix@R=0.8pc@C=2pc{
& X \times Y \ar@{=}[r]
  & X \times Y \ar[dd]^{\inl} \\
X \times D (Y,W) \ar[ur]^{\id \times \eps_{Y,W}}  \ar[dr]_{\eta_{X,Z} \times \id} 
& & \\
& T (X,Z) \times D (Y,W) \ar[r]^-{\psi_{X,Y,Z,W}}
  & X \times Y + Z \times W
}
\]
\[
\small
\hspace*{-2.3cm}
\xymatrix@R=0.8pc@C=1pc{
& T (T (X,Z),Z) \times D (D (Y,W),W) \ar[r]^{\psi_{T(X,Z), D(Y,W),Z,W}}
 & T (X,Z) \times D (Y,W) + Z \times W \ar[r]^{\psi_{X,Y,Z,W} + \id}
   & (X \times Y + Z \times W) + Z \times W \ar[dd]^{\id + \inr} \\
\hspace*{2cm} T (T (X,Z),Z) \times D (Y,W) \ar[ur]^-{\id \times \de_{Y,W}} \ar[dr]_-{\mu_{X,Z} \times \id} \hspace*{-2cm}
& & & \\
& T (X,Z) \times D (Y,W) \ar[rr]^-{\psi_{X,Y,Z,W}}
  & & X \times Y + Z \times W
}
\]

We get a monad-comonad interaction law $(T', D', \psi')$ where
$T' X = \mu Z.\, T(X, Z)$, $D' Y = \nu W.\, D(Y, W)$ and $\psi'$ is
defined as in Section~\ref{sec:functor-functor}. The functors $T'$ and
$D'$ carry monad resp.\ comonad structures \cite{Uus03} and the
natural transformation $\psi$ agrees with those.

\paragraph{Free monad-comonad interaction law}

If $\C$ has relevant initial algebras
and final coalgebras, then, given an interaction law $(F, G, \phi)$,
the free monad-comonad interaction law is provided by the free monad
$F^*$ and the cofree comonad $G^\dagger$ and a suitable natural
transformation $\psi'$.

The free monad is given by $F^* X = \mu Z.\, X + F Z$. Its monad
structure is induced by the parameterized monad $T (X, Z) = X + F Z$. 
Similarly, the cofree comonad is given by
$G^\dagger Y = \nu W.\, Y \times G W$. 
Its comonad structure is
induced by the parameterized comonad $D (Y, W) = Y \times G W$. 
In order to
construct $\psi'$ following the construction we described in the
previous paragraph, we need to construct a family of maps
$\psi_{X,Y,Z,W} : (X + F Z) \times (Y \times G W) \to X \times Y + Z \times W$
natural in $X, Y, Z, W$. This is defined as follows:
\begin{eqnarray*}
\psi_{X,Y,Z,W} & = & 
\xymatrix{
(X + F Z) \times (Y \times G W) \ar[r]^-{\ldist} 
& }
\\
& & \quad
\xymatrix@C=4pc{
X \times (Y \times G W) + F Z \times (Y \times G W) 
   \ar[r]^-{\id \times \fst + \id \times \snd}
& }
\\
& & \quad
\xymatrix@C=3pc{
X \times Y + F Z \times G W \ar[r]^-{\id + \phi_{Z, W}}
& X \times Y + Z \times W
}
\end{eqnarray*}

\paragraph{Restricting to fixed $T$ or $D$} 

We denote the categories obtained from $\mcInt(\C)$ by fixing the
monad $T$ or the comonad $D$ by $\mcInt(\C)|_{T,-}$ and
$\mcInt(\C)|_{-,D}$.
The final object of $\mcInt(\C)|_{T,-}$ is $(T, 0, \psi)$ where
$ 
\psi_{X,Y} = 
\xymatrix@C=0.8pc{
T X \times 0 \ar[r]^-{\snd} 
& 0 \ar[r]^-{?} 
  & X \times Y
}
$; 
note that $0$ is the initial comonad.
The initial object of $\mcInt(\C)|_{-,D}$ is $(\Id, D, \psi)$ where 
$ 
\psi_{X,Y} = 
\xymatrix@C=1.5pc{
X \times D Y \ar[r]^-{\id \times \eps_Y} 
& X \times Y
}
$; 
this is because $\Id$ is the initial monad.

\subsection{Monad-comonad interaction in terms of dual and Sweedler dual}
\label{sec:monad-comonad-d-sd}

Similarly to case of functor-functor interaction laws and maps between
them, the dual allows us to obtain useful alternative
characterizations of monad-comonad interaction laws and their
maps. But a complication arises, see below.\footnote{We discuss these
  isomorphisms of categories only on the level of objects here.}

First, let us notice that we have, canonically, a natural
transformation $\ee : \Id \to \d{\Id}$ and, for any $F$, $G$, a
natural transformation $\mm_{F,G} : \d{F} \cdot \d{G} \to \d{(F \cdot
  G)}$. These are informally defined by $\ee_X x = \lambda_Y.\,
\lambda y.\, (x, y) : X \to \int_Y Y \fun (X \times Y)$ and
$(\mm_{F,G})_X f = \lambda_Y.\, \lambda z.\, \letin{(g, w) \leftarrow
  f_{G Y}\, z}{ g_Y\, w} : \int_{Y'} F Y' \fun (\int_{Y''} G Y'' \fun (X
\times Y'')) \times Y' \to \int_Y F (G Y) \fun (X \times Y)$.  The natural transformation $\ee$ is
a natural isomorphism; its inverse $\ee^{-1} : \d{\Id} \to \Id$ is
defined by $\ee^{-1}_X f = \letin{(x, \_) \leftarrow f 1 \zt}{x} :
\int_Y Y \fun (X \times Y) \to X$.

The data $(\ee, \mm)$ satisfy the conditions to make
$\d{(-)} : [\C, \C]\op \to [\C, \C]$ a lax monoidal functor wrt.\ the
$(\Id, {\cdot})$ composition monoidal structure of $[\C, \C]$.

Now, as a first alternative characterization, a monad-comonad interaction
law of $T$ and $D$ is essentially the same as a natural transformation
$\psia: T \to \d{D}$ satisfying
\[
\small
\xymatrix@R=1.5pc{
\Id \ar[d]_{\eta} \ar[r]^{\ee} 
& \d{\Id} \ar[d]^{\d{\eps}} \\
T \ar[r]^{\psia} 
& 
\d{D}
}
\quad
\xymatrix@R=1.5pc{
T \cdot T \ar[d]_{\mu} \ar[r]^-{\psia \cdot \psia}
& 
\d{D} \cdot \d{D} \ar[r]^-{\mm_{D,D}} 
  & \d{(D \cdot D)} \ar[d]^-{\d{\de}}\\
T \ar[rr]^{\psia}
& & \d{D} 
}
\]

Now, since $\d{(-)} : [\C, \C]\op \to [\C, \C]$ is lax monoidal, it
sends monoids in $[\C, \C]\op$ to monoids in $[\C, \C]$, i.e.,
comonads to monads. In particular, it sends the comonad
$(D, \eps, \de)$ to the monad 
$\d{D} = (\d{D}, \d{\eps} \comp \ee, \d{\de} \comp \mm)$. The
conditions above are precisely the conditions for $\psia$ to be a
monad map from $T$ to $\d{D}$. Summing up, a monad-comonad interaction
law of $T$, $D$ amounts to a monad map $\psia : T \to \d{D}$.

As a second alternative, a monad-comonad interaction law of $T$, $D$
is given by a natural transformation $\psib : D \to \d{T}$ satisfying
\begin{equation} \label{eq:cmd-d-map-conds}
\small
\xymatrix@R=1.5pc{
\Id  \ar[r]^{\ee}
& \d{\Id}  \\
D \ar[u]^{\eps} \ar[r]^{\psib}
& \d{T} \ar[u]_{\d{\eta}}
}
\quad
\xymatrix@R=1.5pc{
D \cdot D \ar[r]^-{\psib \cdot \psib }
& \d{T} \cdot \d{T} \ar[r]^-{\mm_{T,T}}
  & \d{(T \cdot T)}   \\
D \ar[u]^{\de} \ar[rr]^{\psib}
& & \d{T} \ar[u]_{\d{\mu}}
}
\end{equation}

Now, unfortunately, $\d{(-)}$ is not oplax monoidal, so it does
generally not send comonoids to comonoids, and $\d{T}$ is generally
not a comonad.  We could define a candidate counit for $\d{T}$ as
$\ee^{-1} \comp \d{\eta} : \d{T} \to \Id$, but there is generally no
candidate for the comultiplication as we cannot invert $\mm_{T,T}$. So
we cannot generally say that a monad-comonad interaction law is a
comonad map from $D$ to $\d{T}$; the functor $\d{T}$ is not a comonad.

But it may be that there exists what one could informally describe as
the greatest comonad smaller (in an appropriate sense) than
$\d{T}$. The formal object of interest here is what we call, following
the use of this word in other contexts \cite{PS16,Por18,HLFV17}, the Sweedler (or finite) dual of
the monad $T$. It is really just the greatest among all comonads
$D$ satisfying conditions (\ref{eq:cmd-d-map-conds}).

We say that the \emph{Sweedler dual} of the
monad $T$ is the (unique up to isomorphism, if it exists) comonad
$\sd{T} = (\sd{T}, \sd{\eta}, \sd{\mu})$ together with a natural
transformation $\iota : \sd{T} \to \d{T}$ such that
\begin{equation} \label{eq:sdcmd-exi}
\small
\xymatrix@R=1.5pc{
\Id  \ar[r]^{\ee}
& \d{\Id}  \\
\sd{T} \ar[u]^{\sd{\eta}} \ar[r]^{\iota}
& \d{T} \ar[u]_{\d{\eta}}
}
\quad
\xymatrix@R=1.5pc{
\sd{T} \cdot \sd{T} \ar[r]^-{\iota \cdot \iota }
& \d{T} \cdot \d{T} \ar[r]^-{\mm_{T,T}}
  & \d{(T \cdot T)}   \\
\sd{T} \ar[u]^{\sd{\mu}} \ar[rr]^{\iota}
& & \d{T} \ar[u]_{\d{\mu}}
}
\end{equation}
and such that, for any comonad $D = (D, \eps, \de)$ and a natural
transformation $\psib$ satisfying
conditions~(\ref{eq:cmd-d-map-conds}), there exists a unique comonad
map $h : D \to \sd{T}$ satisfying $\psi = \iota \comp h$ as summarized
in the following diagrams:
\[
\small
\xymatrix@R=0.5pc{
& \Id  \ar[r]^{\ee}
  & \d{\Id}  \\
\Id \ar@{=}[ur] \\
& \sd{T} \ar[uu]^{\sd{\eta}} \ar[r]^{\iota}
  & \d{T} \ar[uu]_{\d{\eta}} \\
D \ar[uu]^{\eps} \ar@{.>}[ur]^{h} \ar[urr]_{\psib}
}
\quad
\xymatrix@R=0.5pc{
& \sd{T} \cdot \sd{T} \ar[r]^-{\iota \cdot \iota }
  & \d{T} \cdot \d{T} \ar[r]^-{\mm_{T,T}}
    & \d{(T \cdot T)}   \\
D \cdot D \ar@{.>}[ur]^{h \cdot h}\ar[urr]_{\hspace*{5mm}\psib \cdot \psib}   \\
& \sd{T} \ar[uu]^{\sd{\mu}} \ar[rr]^{\iota}
  & & \d{T} \ar[uu]_{\d{\mu}} \\
D \ar[uu]^{\de} \ar@{.>}[ur]^{h} \ar[urrr]_{\psib}
}
\]
The left-hand diagrams of (\ref{eq:sdcmd-exi}) and
(\ref{eq:cmd-d-map-conds}) are secondary in this definition. In the left-hand
diagram of (\ref{eq:sdcmd-exi}), $\sd{\eta}$ is determined by $\iota$
as $\sd{\eta} = \ee^{-1} \comp \d{\eta} \comp \iota$. The left-hand
diagram of (\ref{eq:cmd-d-map-conds}) commutes trivially when
$\psi = \iota \comp h$ for some comonad map $h$.

Now, if $T$ has the Sweedler dual, there is a bijection between
monad-comonad interaction laws of $T$, $D$, i.e, natural
transformations $\psi : D \to \d{T}$ satisfying
(\ref{eq:cmd-d-map-conds}), and comonad maps $h : D \to
\sd{T}$.
Indeed, any natural transformation $\psi$ satisfying
(\ref{eq:cmd-d-map-conds}) induces a unique comonad map $h$ such that
$\iota \comp h = \psi$ by definition of $\sd{T}$. On the other hand,
for a comonad map $h$, we get a natural transformation $\psi$
satisfying (\ref{eq:cmd-d-map-conds}) simply as the composition
$\iota \comp h$. These constructions are inverses.

To sum up, we have proved that the following categories are isomorphic:
\begin{enumerate}[label=(\roman*)]
\item[(o)] monad-comonad interaction laws;
\item triples of a monad $T$, a comonad $D$ and a monad map from $T$
  to $\d{D}$;
\item triples of a monad $T$, a comonad $D$ and a natural transformation from
  $D$ to $\d{T}$ subject to conditions (\ref{eq:cmd-d-map-conds});
\item triples of a monad $T$, a comonad $D$ and a comonad map from $D$
  to $\sd{T}$.
\end{enumerate}

We see that the initial object of $\mcInt(\C)|_{-,D}$ is
$(\Id, D, \ldots)$ while the final object is $(\d{D}, D, \ldots)$. The
initial object of $\mcInt(\C)|_{T,-}$ is $(T, \sd{T}, \ldots)$ while
the final object is $(T, 0, \ldots)$.

\medskip


Calculating the Sweedler dual is a complicated matter and we will come
to it in Section~\ref{sec:monoid-comonoid}. But here are two examples
where the dual of the underlying functor of a monad is not a comonad
and the underlying functor of the Sweedler dual differs from the dual.

\begin{example} \label{ex:nelists-c}
  In Example~\ref{ex:nelists-a}, we saw that the dual of 
  the functor $T X = X^+$ (nonempty lists) was
  $\d{T} Y \cong \prod n : \Nat.\, [0..n] \times Y$. While the functor
  $T$ is a monad (the free semigroup delivering monad), its dual
  $\d{T}$ is not a comonad. The Sweedler dual is
  $\sd{T} Y = Y \times (Y + Y)$, $\sd{\eta}\, (y, \_) = y$,
  $\sd{\de}\, (y, \inl\, y') = ((y, \inl\, y'), \inl\, (y', \inl\, y'))$,
  $\sd{\de}\, (y, \inr\, y') = ((y, \inr\, y'), \inr\, (y', \inr\, y'))$,
  with 
  $\iota_Y : \sd{T} Y \to \d{T} Y$ defined by \linebreak
  $\iota\, (y, \_)\, 0 = (0, y)$, 
  $\iota\, (\_, \inl\, y')\, (n+1) = (0, y')$, 
  $\iota\, (\_, \inr\, y')\, (n+1) = (n+1, y')$.
  The monad-comonad \linebreak interaction law 
  $\psi_{X,Y} : T X \times \sd{T} Y \to X \times Y$ 
  is defined by 
  $\psi\, ([x_0], (y, \_)) = (x_0, y)$, \linebreak 
  $\psi\, ([x_0, \ldots, x_{n+1}], (\_, \inl\, y')) = (x_0, y')$,
  $\psi\, ([x_0, \ldots, x_{n+1}], (\_, \inr\, y')) = (x_{n+1}, y')$.
\end{example}

\begin{example} \label{ex:update-c} We learned in
  Example~\ref{ex:update-a} that the dual of the functor
  $T X = A \fun (B \times X)$ is
  $\d{T} Y = (A \fun B) \fun (A \times Y)$. But the Sweedler dual of
  $T$ as a monad when $B$ is a monoid acting on $A$ is
  $\sd{T} Y = A \times (B \fun Y)$,
  $\iota\, (a, f) = \lambda g.\, (a, f\, (g\, a))$. 
  In
  Example~\ref{ex:update-b}, we showed the monad-comonad interaction law
  of $T$ and $\sd{T}$.
\end{example}

\section{Stateful running}
\label{sec:running}

Monad-comonad interaction laws are related to stateful runners as
introduced by Uustalu~\cite{Uus15}. Next we present the basic facts
about runners using the Sweedler dual and then explain the connection
to monad-comonad interaction laws.

\subsection{Runners}


A runner is similar to a monad-comonad interaction law but the allowed
machine behaviors are restricted to operate on a fixed state set and
their dynamics is also fixed (in the sense that, for any prospective
initial state, there is a behavior pre-determined). Only the initial
state is not fixed. The state set is manifest but the notion of
machine behavior and the pre-determined dynamics are coalesced with
the interaction protocol into the natural transformation that is the
runner. The runner is a polymorphic function sending any allowed
computation and initial state into a return value and a final state.


\medskip

Given a monad $T = (T, \eta, \mu)$ on $\C$, we call a 
\emph{(stateful) runner} of $T$ an object $Y$ with a family $\theta$ of
maps
\[
\theta_X : T X \times Y \to X \times Y
\]
natural in $X$, satisfying 
\[
\small
\xymatrix@R=1.5pc{
X \times Y \ar@{=}[r]  \ar[d]_{\eta_X \times \id}
& X \times Y \ar@{=}[d] \\
T X \times Y \ar[r]^{\theta_X} 
& X \times Y
}
\quad
\xymatrix@R=1.5pc{
T T X \times Y \ar[d]_{\mu_X \times \id} \ar[r]^{\theta_{T X}}
& T X \times Y \ar[r]^{\theta_X}
  & X \times Y \ar@{=}[d] \\
T X \times Y \ar[rr]^{\theta_X}
& & X \times Y
}
\]

\begin{example}
  We revisit Example~\ref{ex:update-a} about the update monad
  $T X = A \fun (B \times X)$ defined by an action
  ${\downarrow} : A \times B \to A$ of a monoid $B$ on an object
  $A$. An \emph{update lens} \cite{AU14updlens} is an object $Y$
  together with maps $lkp : Y \to A$, $upd : Y \times B \to Y$ such
  that $lkp$ is a map between the $B$-sets $(Y, upd)$ and
  $(A, {\downarrow})$. Any update lens gives us a runner of $T$ via
  $\theta_X : (A \fun (B \times X)) \times Y \to X \times Y$ defined
  by
  $\theta\, (f, y) = \letin{(b, x) \leftarrow f\, (lkp\, y)}{(x, upd\,
    (y, b))}$.
  In fact, runners of this monad are in a bijection with update
  lenses and those in turn are essentially the same as coalgebras for
  the comonad $D Y = A \times (B \fun Y)$.
\end{example}

A \emph{runner map} between $(Y, \theta)$, $(Y', \theta')$ is a map $f
: Y \to Y'$ satisfying $(\id_X \times f) \circ \theta_X = \theta'_{X}
\circ (\id_{T X} \times f)$. 
%
Runners and their maps form a category $\Run(T)$.

Like monad-comonad interaction laws and maps between them, runners and
maps between them admit a number of alternative
characterizations.

The first one is that runners of $T$ are essentially the
same as objects $Y$ endowed with a monad map
$\vartheta : T \to \St^Y$ where $\St^Y = (\St^Y, \eta^Y, \mu^Y)$ is
the state monad for $Y$ whose underlying functor is defined by
$\St^Y X = Y \fun (X \times Y)$.  This is via the bijection of natural
transformations
\[
\mbox{$\int_X$} \C(T X \times Y, X \times Y) \cong \mbox{$\int_X$} \C(T X, \underbrace{Y \fun (X \times Y)}_{\St^Y X})
\]
Under this bijection, the runner conditions amount to the monad
map conditions
\[
\small
\xymatrix@R=1.5pc{
X \ar@{=}[r]  \ar[d]_{\eta_X}
& X \ar[d]^{\eta^Y_X} \\
T X \ar[r]^{\vartheta_X} 
& \St^Y X
}
\quad
\xymatrix@R=1.5pc{
T T X \ar[d]_{\mu_X} \ar[r]^{\vartheta_{T X}}
& \St^Y T X \ar[r]^{\St^Y \vartheta_X}
  & \St^Y \St^Y X \ar[d]^{\mu^Y_X} \\
T X \ar[rr]^{\vartheta_X}
& & \St^Y X
}
\]
A map $f : Y \to Y'$ is a runner map between $(Y, \vartheta)$,
$(Y', \vartheta')$ iff
\[
\small
\xymatrix@R=0.5pc{
T X  \ar[r]^-{\vartheta_X} \ar@{=}[dd]
& Y \fun X \times Y \ar[dr]^{\id \fun \id \times f} \\
& & Y \fun X \times Y' \\
T X \ar[r]^-{\vartheta'_X} 
& Y' \fun X \times Y' \ar[ur]_{f \fun \id}
}
\]

Second, a runner of the monad $T$ is also essentially the same thing
as a coalgebra $(Y, \gamma)$ of the functor $\d{T}$ satisfying the
conditions
\begin{equation} \label{eq:d-coalg-conds}
\small
\xymatrix@R=1.5pc{
Y \ar[r]^{\ee_Y}
& \d{\Id} Y \\
Y \ar@{=}[u] \ar[r]^{\gamma}
& \d{T} Y \ar[u]_{\d{\eta}_Y}
}
\quad
\xymatrix@R=1.5pc{
Y \ar[r]^{\gamma}
& \d{T} Y \ar[r]^{\d{T} \gamma}
  & \d{T} \d{T} Y \ar[r]^{(\mm_{T,T})_Y}
    & \d{(T \cdot T)} Y \\
Y \ar@{=}[u] \ar[rrr]^{\gamma}
& & & \d{T} Y \ar[u]_{\d{\mu}_Y}
}
\end{equation}
This is because of the bijection
\[
\mbox{$\int_X$} \C(T X \times Y, X \times Y) 
\cong \mbox{$\int_X$} \C(Y \times T X, Y \times X)
\cong \C(Y, \underbrace{\mbox{$\int_X$} T X \fun (Y\times X)}_{\d{T} Y})
\]
A runner map between
$(Y, \gamma)$, $(Y', \gamma')$ is a coalgebra map, i.e., a map
$f : Y \to Y'$ such that
\[
\small
\xymatrix@R=1.5pc{
Y \ar[d]_{f} \ar[r]^{\gamma}
& \d{T} Y \ar[d]^{f} \\
Y' \ar[r]^{\gamma'}
& \d{T} Y'
}
\]
%
  Recall that the functor $\d{T}$ is generally not a comonad as
  $\mm_{T,T}$ is not invertible, so we cannot generally speak of
  functor coalgebras satisfying conditions (\ref{eq:d-coalg-conds})  as comonad coalgebras.
%

Lastly, recall that the costate comonad for an object $Y$ is defined
by $\Cost^Y = (\Cost^Y, \eps^Y, \de^Y)$ is defined by 
$\Cost^Y Z = (Y \fun Z) \times Y$, 
$\eps^Y (f, y) = f\ y$, $\de^Y (f, y) = (\lambda
y'.\, (f, y'), y)$.
This gives us a third characterization: a runner is essentially the
same as an object $Y$ together with a natural transformation $\zeta$
between the underlying functor of the costate comonad $\Cost^Y$ and
the functor $\d{T}$ satisfying
\begin{equation} \label{eq:cost-d-map-conds}
\small
\xymatrix@R=1.5pc{
\Id \ar[r]^{\ee}
& \d{\Id} \\
\Cost^Y \ar[u]^{\eps^Y} \ar[r]^{\zeta}
& \d{T} \ar[u]_{\d{\eta}}
}
\quad
\xymatrix@R=1.5pc{
\Cost^Y \cdot \Cost^Y \ar[r]^-{\zeta \cdot \zeta}
& \d{T} \cdot \d{T} \ar[r]^-{\mm_{T,T}}
  & \d{(T \cdot T)} \\
\Cost^Y \ar[u]^{\de^Y} \ar[rr]^{\zeta}
& & \d{T}  \ar[u]_{\d{\mu}}
}
\end{equation}
This is because of the bijection
\[
\C(Y, \d{T} Y) 
\cong^{}\footnote{As $\d{T}$ 
is necessarily strong, we can apply an internal version of 
the Yoneda lemma.}
 \mbox{$\int_Z$} \C(Y \fun Z, Y \fun \d{T} Z)
\cong \mbox{$\int_Z$} \C(\underbrace{(Y \fun Z) \times Y}_{\Cost^Y Z}, \d{T} Z)
\] 
A runner map between $(Y, \zeta)$, $(Y', \zeta')$ is a
map $f : Y \to Y'$ satisfying
\[
\small
  \xymatrix@R=0.5pc{
    & (Y \fun Z) \times Y \ar[r]^-{\zeta'_Z}  & \d{T} Z \ar@{=}[dd] \\
  (Y' \fun Z) \times Y \ar[rd]_-{\id \times f} \ar[ru]^-{(f \fun \id) \times \id} & & \\
    & (Y' \fun Z) \times Y' \ar[r]_-{\zeta_Z} & \d{T} Z
  }
\]

If the Sweedler dual comonad $\sd{T}$ of the monad $T$ exists, then we
can continue this reasoning. We see that a runner is the essentially
the same as an object $Y$ with a comonad morphism between $\Cost^Y$
and $\sd{T}$ and that is further essentially the same as an object $Y$
with a comonad coalgebra of $\sd{T}$.

Summing up, we have established that the following categories are
isomorphic:
\begin{enumerate}[label=(\roman*)]
\item[(o)] runners of $T$;
\item objects $Y$ with a monad map from $T$ to $\St^Y$;
\item functor coalgebras of $\d{T}$ subject to conditions (\ref{eq:d-coalg-conds});
\item objects $Y$ with a natural transformation from $\Cost^Y$ to $\d{T}$ subject to conditions (\ref{eq:cost-d-map-conds});
\item objects $Y$ with a comonad map from $\Cost^Y$ to $\sd{T}$;
\item comonad coalgebras of $\sd{T}$.
\end{enumerate}

\subsection{Runners vs.\ monad-comonad interaction laws}

Monad-comonad interaction laws of $T$, $D$ are in a bijection with
\emph{$D$-coalgebraic $T$-runner specs} by which we mean
carrier-preserving functors between $\Coalg(D)$ and $\Run(T)$, i.e.,
functors $\Psi : \Coalg(D) \to \Run(T)$ such that
\[
\small
\xymatrix@R=1.5pc{
\Coalg(D) \ar[dr]_{U} \ar[rr]^{\Psi}
& & \Run(T) \ar[dl]^{U} \\
& \C
}
\]

Indeed, given a monad-comonad interaction law $\psi$, we can define a
runner spec $\Psi$ by
\[
\small
(\Psi\, (Y, \gamma))_X = (Y, 
\xymatrix{
T X \times Y \ar[r]^-{\id \times \gamma}
& T X \times D Y \ar[r]^-{\psi_{X,Y}}
  & X \times Y 
} )
\]
In the opposite direction, given a runner spec $\Psi$, we build a
interaction law from the cofree coalgebras of $D$. For any $Y$, we
have the cofree coalgebra $(D Y, \de_Y)$ and define a monad-comonad
interaction law $\phi$ by
\[
\small
\phi_{X,Y} =
\xymatrix@C=3.5pc{
T X \times D Y \ar[r]^{\Psi (DY,\de_Y)_X}
& X \times D Y \ar[r]^{\id \times \eps_Y} 
  & X \times Y
}
\]

A pair of a monad map $f : T \to T'$ and a comonad map $g : D' \to D$
is an interaction law map between $(T, D, \psi)$ and $(T', D', \psi')$
iff the corresponding coalgebraic runner specs
satisfy
\[
\small
\xymatrix@R=1.5pc{
\Coalg(D) \ar[r]^{\Psi}
& \Run(T)  \\
\Coalg(D')  \ar[u]^{\Coalg(g)} \ar[r]^{\Psi'}
& \Run(T') \ar[u]_{\Run(f)}
}
\]
(Notice that $\Coalg(-) : \Comnd(\C) \to \CAT$ and
$\Run(-) : (\Mnd(\C))\op \to \CAT$.) So the categories of
monad-comonad interaction laws and coalgebraic runner specs are
isomorphic.

More modularly, but assuming that all Sweedler duals exist, the
isomorphism of the categories of monad-comonad interaction laws and
coalgebraic runner specs follows from the following sequence of
isomorphisms of categories, using that $\Run(T) \cong \Coalg(\sd{T})$:
\begin{enumerate}[label=(\roman*)]
\item[(o)] monad-comonad interaction laws;
\item triples of a monad $T$, a comonad $D$ and a comonad map
  between $D$, $\sd{T}$;
\item triples of a monad $T$, a comonad $D$ and a carrier-preserving
  functor between $\Coalg(D)$, $\Coalg(\sd{T})$;
\item coalgebraic runner specs.
\end{enumerate}

\section{Residual interaction and running}
\label{sec:residual}

We will now generalize interaction laws to allow that that not all of
the effect of a computation is serviced by a machine behavior in an
interaction.

\subsection{Residual interaction}

Given a monad $R = (R, \eta^R, \mu^R)$ on our base category $\C$. We
can generalize functor-functor and monad-comonad interaction laws as
follows.

An \emph{$R$-residual functor-functor interaction law} is given by
endofunctors $F$, $G$ on $\C$ together with a family of maps 
\[
\phi_{X,Y} : F X \times G Y \to R (X \times Y)
\] 
natural in $X$, $Y$.

An \emph{$R$-residual interaction law map} between $(F, G, \phi)$,
$(F', G', \phi')$ is given by natural transformations $f : F \to F'$,
$g : G' \to G$ such that
\[
\small
\xymatrix@R=0.4pc{
& F X \times G Y \ar[r]^-{\phi_{X,Y}}
  & R(X \times Y) \ar@{=}[dd] \\
F X \times G' Y \ar[ur]^{\id \times g_Y}  \ar[dr]_{f_X \times \id} 
& & \\
& F' X \times G' Y \ar[r]^-{\phi'_{X,Y}}
  & R(X \times Y)
}
\]

$R$-residual functor-functor interaction laws form a category
$\Int(\C, R)$.

This category is monoidal.
The tensorial unit is $(\Id, \Id, \eta^R \cdot (\Id \times \Id))$.
The tensor of $(F, G, \phi)$ and $(J, K, \psi)$ is $(F \cdot J, G
\cdot K, \mu^R \comp R \cdot \psi \comp \phi \cdot (J \times K))$.

An $R$-\emph{residual monad-comonad interaction law}
of a monad $T$ and a comonad $D$ is a family $\psi$ of maps
\[
\psi_{X,Y} : T X \times D Y \to R (X \times Y)
\]
natural in $X$ and $Y$, satisfying
\[
\small
\xymatrix@R=0.4pc{
& X \times Y \ar@{=}[r]
  & X \times Y \ar[dd]^{\eta^R_{X \times Y}} \\
X \times D Y \ar[ur]^{\id \times \eps_Y}  \ar[dr]_{\eta_X \times \id} 
& & \\
& T X \times D Y \ar[r]^-{\psi_{X,Y}}
  & R (X \times Y) 
}
\]
\[
\xymatrix@R=0.4pc{
& T T X \times D D Y \ar[r]^{\psi_{TX, DY}}
 & R (T X \times D Y) \ar[r]^{R \psi_{X,Y}}
   & R R (X \times Y) \ar[dd]^{\mu^R_{X \times Y}} \\
T T X \times D Y \ar[ur]^{\id \times \de_Y} \ar[dr]_{\mu_X \times \id} 
& & & \\
& T X \times D Y \ar[rr]^-{\psi_{X,Y}}
  & & R (X \times Y)
}
\]

\begin{example}
  Let $R X = X + E$ (the exceptions monad). Take
  $T X = A \fun (X + E)$,  
  $D Y = A \times Y$; these are a monad and a comonad. The
  natural transformation
  $\psi\, (f, (a, y)) = \mathsf{case~} f\, a \mathsf{~of~} (\inl\, x
  \mapsto (\inl\, x, y) \mid \inr\, e \mapsto \inr\, e)$
  satisfies the conditions of a $R$-residual monad-comonad interaction
  law.

\end{example}

$R$-residual monad-comonad interaction laws are the same as monoid
objects in the monoidal category $\Int(\C, R)$. 

$R$-\emph{residual monad-comonad interaction law maps} are defined as
expected and correspond to monoid morphisms.

The category $\mcInt(\C, R)$ of $R$-residual monad-comonad interaction
laws is isomorphic to \linebreak $\Mon(\Int(\C, R))$.

\subsection{Relationship to interaction laws on Kleisli
  categories}

It is tempting to guess that an $R$-residual functor-functor
interaction law of $F$, $G$ would be the same thing as a
functor-functor interaction law on the Kleisli category of $R$. But
this is jumping to conclusions too hastily. For something like this to
be feasible, we need, first of all, that $F$, $G$ lift to $\Kl(R)$. A
necessary and sufficient condition is the presence of distributive
laws of $F$ and $G$ over $R$, i.e., natural transformations
$\kappa : F \cdot R \to R \cdot F$ and
$\lambda : G \cdot R \to R \cdot G$ agreeing with 
 the monad structure of
$R$. Then we define the lifted versions of $F$, $G$ on objects by
$\bar{F} X = F X$, $\bar{G} Y = G Y$; for maps $k : X \to R X'$,
$\ell : Y \to R Y'$, we define
$\bar{F} k = \kappa_{X'} \comp F k : F X \to R F X'$ and
$\bar{G} \ell = \lambda_{Y'} \comp G \ell : G Y \to R G Y'$.

Moreover, we also need to lift $\times$ to $\Kl(R)$ as a bifunctor and
monoidal structure. For this, a necessary and sufficient condition is
monoidality of $R$ as a monad, i.e., the presence of a family of maps
$m_{X,Y} : R X \times R Y \to R (X \times Y)$ 
natural in $X$, $Y$
agreeing with both the product monoidal structure of $\C$ and the
monad structure of $R$.  (This is the same as $R$ being commutative
strong monad.) For objects, we then define
$X \mathbin{\bar{\times}} Y = X \times Y$, and for maps
$k : X \to R X'$, $\ell : Y \to R Y'$, we define
$k \mathbin{\bar{\times}} \ell = m_{X',Y'} \comp (k \times \ell) : X \times Y \to
R (X' \times Y')$.

The naturality condition for
$\phi_{X,Y} : F X \times G Y \to R (X \times Y)$ as an interaction law
of $\bar{F}$, $\bar{G}$ is: for all $k$, $\ell$,
\[
\small
\xymatrix@R=1.5pc{
F X \times G Y \ar[rr]^{\phi_{X, Y}} \ar[d]_{F k \times G \ell}
& & R (X \times Y) \ar[d]^{R (k \times \ell)} \\
F R X' \times G R Y' \ar[d]_{\kappa_{X'} \times \lambda_{Y'}}
& & R (R X' \times R Y') \ar[d]^{R m^R_{X',Y'}}   \\
R F X' \times R G Y' \ar[d]_{m^R_{FX', GY'}}
& & R R (X' \times Y') \ar[d]^{\mu^R_{X' \times Y'}} \\
R (F X' \times G Y') \ar[r]^{R \phi_{X',Y'}} 
& R R (X' \times Y') \ar[r]^{\mu^R_{X' \times Y'}}
  & R (X' \times Y')
}
\]

But the naturality condition for $\phi$ as an
$R$-residual interaction law of $F$, $G$ is: for all $f$, $g$,
\[
\small
\xymatrix@R=1.5pc{
F X \times G Y \ar[r]^{\phi_{X,Y}} \ar[d]_{F f \times G g}
& R (X \times Y) \ar[d]^{R (f \times g)} \\
F X' \times G Y' \ar[r]^{\phi_{X',Y'}} 
& R (X' \times Y') \\
}
\]

The first condition implies the second:
\[
\small
\xymatrix@C=3pc{
F X \times G Y \ar[rr]^{\phi_{X, Y}} \ar[d]_{F f \times G g}
& & R (X \times Y) \ar[d]^{R (f \times g)}  \\
F X' \times G Y' \ar[d]^{F \eta^R_{X'} \times G \eta^R_{Y'}}
     \ar@/_3pc/[dd]_{\eta^R_{FX'} \times \eta^R_{GY'}}
     \ar@/_5pc/[ddd]_{\eta^R_{FX' \times GY'}}
     \ar@/_8pc/[dddd]_{\phi_{X',Y'}}
& & R (X \times Y) \ar[d]_{R (\eta^R_{X'} \times \eta^R_{Y'})} 
     \ar@/^3pc/[dd]^{R\eta^R_{X' \times Y'}}
     \ar@{=}@/^5pc/[ddd]
\\
F R X' \times G R Y' \ar[d]^{\kappa_{X'} \times \lambda_{Y'}}
& & R (R X' \times R Y') \ar[d]_{R m^R_{X',Y'}}   \\
R F X' \times R G Y' \ar[d]^{m^R_{FX', GY'}}
& & R R (X' \times Y') \ar[d]_{\mu^R_{X' \times Y'}} \\
R (F X' \times G Y') \ar[r]^{R \phi_{X',Y'}} 
& R R (X' \times Y') \ar[r]^{\mu^R_{X' \times Y'}}
  & R (X' \times Y') \\
R (X' \times Y') \ar[ur]^{\eta^R_{R (X' \times Y')}}
  \ar@{=}[urr]
}
\]

The second condition gives the first condition restricted to pure maps
of $\Kl(R)$ (maps in the image of the left adjoint $J$ the Kleisli
adjunction of $R$), i.e., for maps $k$, $\ell$ of the form
$k = J f = \eta_{X'} \comp f$, $\ell = J g = \eta_{Y'} \comp g$:
\[
\small
\xymatrix@C=3pc{
F X \times G Y \ar[rr]^{\phi_{X, Y}} \ar[d]_{F f \times G g}
                       \ar@/_5pc/[dd]_{F k \times G \ell}
& & R (X \times Y) \ar[d]^{R (f \times g)} \ar@/^5pc/[dd]^{R (k \times \ell)}  \\
F X' \times G Y' \ar[rr]^{\phi_{X', Y'}} \ar[d]_{F \eta^R_{X'} \times G \eta^R_{Y'}}
     \ar@/^3pc/[dd]^{\eta^R_{FX'} \times \eta^R_{GY'}}
     \ar@/^5pc/[ddd]^{\eta^R_{FX' \times GY'}}
& & R (X \times Y) \ar[d]^{R (\eta^R_{X'} \times \eta^R_{Y'})} 
     \ar@/_3pc/[dd]_{R\eta^R_{X' \times Y'}}
     \ar@{=}@/_5pc/[ddd]
     \ar@/_4pc/[dddl]_{\eta^R_{R(X' \times Y')}}
\\
F R X' \times G R Y' \ar[d]_{\kappa_{X'} \times \lambda_{Y'}}
& & R (R X' \times R Y') \ar[d]^{R m^R_{X',Y'}}   \\
R F X' \times R G Y' \ar[d]_{m^R_{FX', GY'}}
& & R R (X' \times Y') \ar[d]^{\mu^R_{X' \times Y'}} \\
R (F X' \times G Y') \ar[r]^{R \phi_{X',Y'}} 
& R R (X' \times Y') \ar[r]^{\mu^R_{X' \times Y'}}
  & R (X' \times Y')
}
\]

We thus see that $R$-residual functor-functor interaction laws are
more liberal than functor-functor interaction laws in $\Kl(R)$ in that
we do not need the distributive laws and monoidality of $R$ and that
the naturality condition is weaker (only required for pure maps).

\subsection{Residual stateful running}

Similarly to interaction laws, the concept of runners can also be
generalized.

Given a monad $R = (R, \eta^R, \mu^R)$ on $\C$.  An \emph{$R$-residual
  runner} of a monad $T = (T, \eta, \mu)$ on $\C$ is an an object $Y$
with a family $\theta$ of maps
\[
\theta_X : T X \times Y \to R (X \times Y)
\]
natural in $X$, satisfying
\[
\small
\xymatrix@R=1.5pc{
X \times Y \ar@{=}[r]  \ar[d]_{\eta_X \times \id}
& X \times Y \ar[d]^{\eta^R_{X \times Y}} \\
T X \times Y \ar[r]^{\theta_X} 
& R (X \times Y)
}
\xymatrix@R=1.5pc{
T T X \times Y \ar[d]_{\mu_X \times \id} \ar[r]^{\theta_{T X}}
& R(T X \times Y) \ar[r]^{R\theta_X}
  & RR (X \times Y) \ar[d]^{\mu^R_{X \times Y}} \\
T X \times Y \ar[rr]^{\theta_X}
& & R (X \times Y)
}
\]

A \emph{map of $R$-residual runners} of $T$ between $(Y, \theta)$,
$(Y', \theta')$ is a map $f : Y \to Y'$ satisfying
\[
\small
\xymatrix@R=1.5pc{
T X \times Y \ar[r]^{\theta_X} \ar[d]_{\id \times f}
& R (X \times Y) \ar[d]^{R (\id \times f)} \\
T X \times Y' \ar[r]^{\theta'_X} 
& R (X \times Y')
}
\]

$R$-residual runners of $T$ form a category $\Run(T, R)$.

$R$-residual runners of $T$ are essentially the same as objects $Y$
endowed with a monad map $\vartheta : T \to \St^{R,Y}$ where
$\St^{R,Y} = (\St^{R,Y}, \eta^{R,Y}, \mu^{R,Y})$ is the
$R$-transformed state monad for $Y$ whose underlying functor is
defined by $\St^{R,Y} X = Y \fun R(X \times Y)$.

This is via the bijection of natural transformations
\[
\mbox{$\int_X$} \C(T X \times Y, R(X \times Y))
\cong \mbox{$\int_X$} \C(T X, \underbrace{Y \fun R(X \times Y)}_{\St^{R,Y} X})
\]
Under this bijection, the $R$-residual runner conditions amount to the monad
map conditions
\[
\small
\xymatrix@R=1.5pc@C=2pc{
X \ar@{=}[r]  \ar[d]_{\eta_X}
& X \ar[d]^{\eta^{R,Y}_X} \\
T X \ar[r]^-{\vartheta_X} 
& \St^{R,Y} X
}
\quad
\xymatrix@R=1.5pc@C=2pc{
T T X \ar[d]_{\mu_X} \ar[r]^-{\vartheta_{T X}}
& \St^{R,Y} T X \ar[r]^-{\St^{R,Y} \vartheta_X}
  & \St^{R,Y} \St^{R,Y} X \ar[d]^{\mu^{R,Y}_X} \\
T X \ar[rr]^-{\vartheta_X}
& & \St^{R,Y} X
}
\]

A map $f : Y \to Y'$ is a map of $R$-residual runners of $T$ between
$(Y, \vartheta)$, $(Y', \vartheta')$ iff
\[
\small
\xymatrix@R=0.4pc{
T X  \ar[r]^-{\vartheta_X} \ar@{=}[dd]
& Y \fun R(X \times Y) \ar[dr]^{\id \fun R(\id \times f)} \\
& & Y \fun R(X \times Y') \\
T X \ar[r]^-{\vartheta'_X} 
& Y' \fun R(X \times Y') \ar[ur]_{f \fun \id}
}
\]
So the categories of $R$-residual runners of $T$ and objects $Y$ equipped
with a monad map from $T$ to $\St^{R,Y}$ are isomorphic.

\section{Monoid-comonoid interaction}
\label{sec:monoid-comonoid}

Exploiting that monads and monad-like objects like arrows or lax
monoidal functors (``applicative functors'') are monoids has turned
out to be very fruitful in categorical semantics (see, e.g.,
\cite{JM10,CK14,RJ17}). We now explore this perspective by
abstracting monad-comonad interaction laws into monoid-comonoid
interaction laws. This leads us to further known concepts and methods
from category theory.

\subsection{Interaction laws and Chu spaces}

The first step in generalizing interaction laws to monoids and comonoids
is to account for interaction laws as maps in a category. Recall
that the Day convolution~\cite{Da:1970} of functors
$F, G : \C \to \C$ where $\C$ is a category with finite products is given by
\[
(F \star G) Z = \mbox{$\int^{X, Y}$} \C(X \times Y, Z) \bullet (F X \times G Y)
\]
provided that this coend exists. (We take the same 
stance toward the question of well-definedness of the Day convolution
as we took toward the well-definedness of the dual in Section~\ref{sec:functor-functor}.) 
By reasoning about natural transformations, we see that interaction laws
for a pair of functors $F$ and $G$ amount to maps $\phi : F
\star G \to \Id_\C$:
\[
\mbox{$\int_{X,Y}$} \C(F X \times G Y, X \times Y)
\cong \mbox{$\int_{X,Y,Z}$} \Set(\C(X \times Y, Z), \C(F X\times G Y, Z))
\cong \mbox{$\int_Z$} \C((F \star G) Z, Z) 
\]
%
%
%
We see that a functor-functor interaction law is a triple $\left(F, G,
\phi : F \star G \to \Id_\C\right)$, i.e., a Chu space~\cite{Bar06} over the monoid object $\Id_\C$ wrt. the
Day convolution monoidal structure on $[\C,\C]$. An
interaction law map is a Chu space map under this view, 
so the category $\Int(\C)$ is
isomorphic to the category $\mathbf{Chu}([\C,\C], \Id_\C)$.

This is nice, but not fine-grained enough for developing an
abstract foundation for our theory. The canonical monoidal
structure on $\mathbf{Chu}(\F, R)$ (where $R$ is a monoid object in $\F$) 
is based on the monoidal structure
of the base category $\F$, which in our case is the Day convolution, and
uses pullbacks. But we are interested in a different monoidal
structure on $\Int(\C)$ that is based on composition and gives us
monads and comonads as monoids resp.\ comonoids. We fix this mismatch by moving to
one of the cousins of the Chu construction: glueing \`a la Hasegawa.

\subsection{Interaction laws and Hasegawa's glueing}

Hasegawa's glueing construction \cite{Has99} works as follows. Given
two monoidal categories $\F = (\F, I^\F, \otimes^F)$, $\G = (\G, I^\G,
\otimes^\G)$ and a lax monoidal functor $(\d{(-)}, \ee, \mm) : \G \to
\F$. The comma category $\F \downarrow \d{(-)}$
carries a monoidal structure given by:
\[
\begin{array}{c}
I = (I^\F, I^\G, 
\xymatrix@C=1pc{
I^\F \ar[r]^-{\ee} 
& \d{(I^\G)}
}
) \\
(F, G, \phi) \otimes (F', G', \phi') = 
(F \otimes^\F F', G \otimes^\G G', 
\xymatrix@C=1.5pc{F \otimes^\F F' \ar[r]^-{\phi \otimes^\F\! \phi'}
& \d{G} \otimes^\F \d{G'} \ar[r]^-{\mm_{G,G'}}
  & \d{(G \otimes^\G G')}    
}
)
\end{array}
\]
Also, if $\F$ and $\G$ are closed and $\G$ has pullbacks, then
$\F \downarrow \d{(-)}$ is closed.

An interesting case of this construction is when we start with a
duoidal category $(\F, I, \otimes, J, \star)$ closed
wrt.\ $\star$~\cite{AM:2010,GLF16}. This is a category with two monoidal
structures, and among its data are a map $\chi : I \star I \to I$ and
a family of maps $\xi_{F,F',G,G'} : (F \otimes F') \star (G \otimes
G') \to (F \star G) \otimes (F' \star G')$ natural in $F, F', G,
G'$. Moreover, given a monoid $(R, \eta^R, \mu^R)$ in $(\F, I,
\otimes)$, we define $\d{(-)} : \F^{\op} \to \F$ by $\d{G} = G
\lollistar R$. This functor $\d{(-)}$ is lax monoidal wrt. the 
$(I, \otimes)$ monoidal structure, since as witnesses $\ee, \mm$ of
lax monoidality we have the curryings of 
\[
\begin{array}{c}
\ee' = 
\xymatrix{
I \star I \ar[r]^{\chi} 
& I \ar[r]^{\eta^R}
  & R
}
\\
\mm'_{G,G'} = 
\xymatrix{
(\d{G} \otimes \d{G'}) \star (G \otimes G') \ar[r]^-{\xi}
& (\d{G} \star G) \otimes (\d{G'} \star G') \ar[r]^-{\ev \otimes \ev}
  & R \otimes R \ar[r]^-{\mu^R}
    & R
}
\end{array}
\]
Since $\F$, $\F^{\op}$, $\d{(-)}$ fulfill the assumptions of the
glueing construction, $\F \dn \d{(-)}$ is monoidal. A monoid in this
category is what we will recognize as a $R$-residual monoid-comonoid interaction
law in the duoidal category $\F$.

To recover the usual functor-functor and monad-comonad interaction
laws, we take $(\F, I, \otimes, J, \star, \lollistar)$
to be $[\C, \C]$ with its composition monoidal and Day convolution
monoidal closed structures, and
define $\d{G} = G \lollistar \Id$. An object of $\F \dn \d{(-)}$ is a
functor-functor interaction law while a monad-comonad interaction law
is a monoid object of this category. We ignore the issue that $\star$
and $\lollistar$ need not be well-defined everywhere on $[\C,\C]$.  As
we remarked before, this can be solved by restricting to a full
subcategory of $[\C, \C]$ given by some class of functors that is
closed under $\star$ and $\lollistar$ (such as finitary
functors, cf.~\cite{GLF16}).


The notions of dual and Sweedler dual emerge as follows in this
setting. When the $\star$ monoidal structure is symmetric, we also
have 
\[
\xymatrix@C=3pc{
\F\op \ar@/^0.6pc/[r]^{\d{(-)}}  \ar@{}[r]|{\top}
& \F \ar@/^0.6pc/[l]^{{\d{(-)}}\op}
}
\]
since, for any $F, G \in |[\C, \C]|$,
\[
\F(F, \d{G}) \cong \F(F \star G, \Id) \cong \F(G \star F, \Id) \cong \F(G,  \d{F}) \cong \F\op (\d{F},  G).
\]
Because of this adjunction, we call $\d{(-)}$ the \emph{dual}.

Since the functor $\d{(-)}$ is lax monoidal, it lifts to a functor
between the respective categories of monoids (bear in mind that
$\Mon(\F\op) = (\Comon(\F))\op$):
\begin{equation} \label{eq:adj-mon}
\xymatrix@R=1.5pc{
(\Comon(\F))\op \ar[r]^-{\d{(-)}} \ar[d]_{U}
& \Mon(\F) \ar[d]^{U} \\
\F\op \ar[r]^-{\d{(-)}}
& \F 
}
\end{equation}

However, its left adjoint ${\d{(-)}}\op$ is only oplax monoidal, but
not lax monoidal, so we cannot get a similar diagram for
${\d{(-)}}\op$. We want to find a substitute for this lifting, in
particular, we want a left adjoint for the lifted $\d{(-)}$:
\[
\xymatrix@C=4pc{
(\Comon(\F))\op \ar@/^0.8pc/[r]^-{\d{(-)}}  \ar@{}[r]|-{\top}
& \Mon(\F) \ar@/^0.8pc/[l]^-{{\sd{(-)}}\op}
}
\]
We obtain not a natural isomorphism between two functors $(\Comon(\F))\op \to \F$ as in diagram~(\ref{eq:adj-mon}), but
instead only a natural transformation $\iota :{ \d{(-)}}\op \cdot U \to U
\cdot {{\sd{(-)}}\op}$ between two functors $\Mon(\F) \to \F\op$.
\[
\xymatrix@R=1.5pc{
(\Comon(\F))\op \ar[d]_{U} 
& \Mon(\F) \ar[l]_-{{\sd{(-)}}\op} \ar[d]^{U} \xtwocell[ld]{}<>{^\iota} \\
\F\op \
& \F \ar[l]^-{{\d{(-)}}\op}
}
\]
We call the functor $\sd{(-)} : (\Mon(F))\op \to \Comon(F)$ the
\emph{Sweedler dual}.

\subsection{Sweedler dual for some constructions of monoids}

As we have seen in the setting of monad-comonad interaction laws, it
is not always easy to find the Sweedler dual. In the remainder of
this section, we focus on the cases of free monoids and free monoids
quotiented by ``equations'' for one method to compute them.

Let $F^*$ be the free monoid on $F$. In this case, if the cofree
comonoid on $\d{F}$ exists, then it is the Sweedler dual of $F^*$,
i.e., we can show that $\sd{(F^*)} = (\d{F})^\dagger$. This is seen
from the following calculation:
\[
\renewcommand{\arraystretch}{1.5}
\begin{array}{l}
(\Comon(\F))\op((\d{F})^\dagger, D)
    \cong \Comon(\F)(D, (\d{F})^\dagger) 
    \cong \F(U D, \d{F}) \\
\quad
    \cong \F\op(\d{F}, U D)
    \cong \F(F, \d{(U D)}) 
    \cong \F(F, U \d{D}) 
    \cong \Mon(\F)(F^*, \d{D})
\end{array} 
\]
This observation facilitates calculation of Sweedler duals of free
monads (i.e., theories without equations). A natural question is to
ask what happens in the presence of equations. Suppose that we
have a monoid $T$ given as a coequalizer 
\[
\xymatrix@C+=2cm{
  E^* \ar@<.5ex>[r]^-{f^L} \ar@<-.5ex>[r]_-{g^L} & F^* \ar[r] & T
}
\]
in $\Mon(\F)$ where $(-)^L$ is the left transpose of the
free/forgetful adjunction between $\F$ and $\Mon(\F)$. The maps
$f, g : E \to U F^*$ of $\F$ represent a system of equations in a set
of variables $E$, and we can think of $T$ as being the monoid obtained
by calculating the free monoid and then quotienting by the
equations. We can try to obtain the Sweedler dual of $T$ by
constructing a ``dual'' diagram as follows. We can instantiate $\iota$
at $F^*$ and obtain a map
$\iota_{F^*} : \d{(U {F^*})} \to U
(\sd{({F^*})})$
in $\F\op$, i.e., a map
$\iota_{F^*} : U ({\sd{({F^*})}}) \to {
  \d{(U {F^*})}}$
in $\F$. By composing with $\d{f}$ and $\d{g}$, we get:
\[
\xymatrix@C+=2cm{
  U ((\d{F})^\dagger) = U (\sd{(F^*)}) \ar[r]^-{\iota_{F^*}} & \d{(U F^*)} \ar@<.5ex>[r]^-{\d{f}} \ar@<-.5ex>[r]_-{\d{g}} & \d{E}
}
\]
The Sweedler dual $\sd{T}$ of $T$ is now obtained as an equalizer in
$\Comon(\F)$ by 
\[
\xymatrix@C+=2cm{
\sd{T} \ar[r] &
  (\d{F})^\dagger \ar@<.5ex>[r]^-{(\d{f} \comp \iota_{F^*})^R} \ar@<-.5ex>[r]_-{(\d{g} \circ \iota_{F^*})^R} & (\d{E})^\dagger
}
\]
where $(-)^R$ is the right transpose of the
forgetful/cofree adjunction between $\Comon(\F)$ and $\F$.

\begin{example} \label{ex:nelists-d} Revisiting
  Example~\ref{ex:nelists-c}, the nonempty list monad $T X = X^+$ arises
  as the quotient of the free monad $T_0 X = \mu Z.\, X + Z \times Z$
  by the associativity equation for its operation
  $c_X : X \times X \to T_0 X$, i.e., the equation
\[
\small
\xymatrix@R=0.5pc{
(X \times X) \times X \ar[dd]_{\ass} \ar[r]^-{c_X \times \eta_X}
& T_0 X \times T_0 X \ar[r]^-{c_{T_0X}} 
  & T_0 T_0 X \ar[dr]^{\mu_X} \\
& & & T_0 X \\
X \times (X \times X) \ar[r]^-{\eta_X \times c_X}
& T_0 X \times T_0 X \ar[r]^-{c_{T_0X}} 
  & T_0 T_0 X \ar[ur]_{\mu_X}    
}
\]
The monad $T_0$ is the free monad on the functor $F X = X \times X$. The dual of $F$ is $G Y = Y + Y$.
The Sweedler dual of $T$ is the subcomonad of the cofree
  comonad $\sd{T_0} Y = \nu W.\, Y \times (W + W)$ by the
  coassociativity coequation for its cooperation
  $c'_Y : \sd{T_0} Y \to Y + Y$,
 i.e., the coequation
\[
\small
\xymatrix@R=0.5pc{
& \sd{T_0} \sd{T_0} Y \ar[r]^-{c'_{\sd{T_0} Y}}
  & \sd{T_0} Y + \sd{T_0} Y \ar[r]^-{c'_Y + \eps_Y}
    & (Y + Y) + Y \ar[dd]^{\ass} \\
\sd{T_0} Y \ar[ur]^-{\de_Y} \ar[dr]_-{\de_Y} \\
& \sd{T_0} \sd{T_0} Y \ar[r]^-{c'_{\sd{T_0} X}}
  & \sd{T_0} Y + \sd{T_0} Y \ar[r]^-{\eps_Y + c'_Y}
    & Y + (Y + Y)
}
\] 
With some calculation, we can find that
$\sd{T} Y \cong Y \times (Y + Y)$.\footnote{This calculation was carried out in detail in \cite{Uus15}.} The comonad map $i : \sd{T} \to \sd{T_0}$ is defined by 
$i\, (y, \inl\, y') = (y, \inl\, (i\, (y', \inl\, y')))$,
$i\, (y, \inr\, y') = (y, \inr\, (i\, (y', \inr\, y')))$,

Compared to $\sd{T_0}$, the comonad $\sd{T}$ is relatively degenerate
because coassociativity entails corectangularity (while associativity
does not entail rectangularity\footnote{In band theory, left and
  right rectangularity are the equations $(x*y)*z = x*z$ and $x*(y*z)
  = x*z$.}), as the following theorem shows.
\begin{theorem}
  Given a comonad $(D, \eps,
\de)$ on $\C$ with a cooperation $c_Y : D Y \to Y + Y$. We show that
the coequation of coassociativity
\[
\small
\xymatrix@R=0.5pc{
& D D Y \ar[r]^-{c_{D Y}}
  & D Y + D Y \ar[r]^-{c_Y + \eps_Y}
    & (Y + Y) + Y \ar[dd]^{\ass} \\
D Y \ar[ur]^-{\de_Y} \ar[dr]_-{\de_Y} \\
& D D Y \ar[r]^-{c_{D Y}}
  & D Y + D Y \ar[r]^-{\eps_Y + c_Y}
    & Y + (Y + Y)
}
\]
implies left and right corectangularity, i.e. the two coequations
\[
\small
\xymatrix@R=1.5pc{
& D D Y \ar[r]^-{c_{D Y}}
  & D Y + D Y \ar[r]^-{c_Y + \eps_Y}
    & (Y + Y) + Y  \\
D Y \ar[ur]^{\de_Y} \ar[dr]_{\de_Y} \ar[rrr]^{c_Y} 
& 
  & 
    & Y + Y \ar[u]_-{\inl + \id} \ar[d]^{\id + \inr} \\
& D D Y \ar[r]^-{c_{D Y}}
  & D Y + D Y \ar[r]^-{\eps_Y + c_Y}
    & Y + (Y + Y)
}
\]
\end{theorem}
\begin{proof}
We can pull $c_{D Y} \comp \delta_Y$ back along the coproduct
coprojections (the existence of these pullbacks is part of
extensivity):
\[
\small
\xymatrix@R=1.5pc{
P Y \ar[rr]^-{f_Y} \ar[d]_{i_Y}
& & D Y \ar[d]^{\inl} \\
D Y \ar[r]^-{\de_Y}
& D D Y \ar[r]^-{c_{DY}}
  & D Y + D Y \\
Q Y \ar[rr]^-{g_Y} \ar[u]^{j_Y}
& & D Y \ar[u]_{\inr}
}
\]
By stability of coproducts under pullback (which is also part of
extensivity), $(D Y, i_Y, j_Y)$ is a coproduct $P Y$ and $Q Y$.

For right corectangularity, we notice that the two maps 
$(\eps_Y + c_Y) \comp c_{D Y} \comp \delta_Y$ and
$(Y + \inr) \comp (\eps_Y + \eps_Y) \comp c_{D Y} \comp \delta_Y$ 
both satisfy both triangles of the unique copair of
$\inl \comp \eps_Y \comp f_Y$ and $\inr \comp c_Y \comp g_Y$, so they
must be the same map. Indeed, we have both
\[
\small
\xymatrix@R=1.5pc{
P Y \ar[rr]^-{f_Y} \ar[d]_{i_Y}
& & D Y \ar[r]^-{\eps_Y} \ar[d]^{\inl}  
    & Y \ar[d]^{\inl} \\
D Y \ar[r]^-{\de_Y}
& D D Y \ar[r]^-{c_{DY}}
  & D Y + D Y \ar[r]^-{\eps_Y + c_Y}
    & Y + (Y + Y) \\
Q Y \ar[rr]^-{g_Y} \ar[u]^{j_Y}
& & D Y \ar[u]_{\inr} \ar[r]^-{c_Y}
    & Y + Y \ar[u]_{\inr}
}
\]
and, using coassociativity, also 
\[
\small
\xymatrix@R=1.2pc{
P Y \ar[rr]^-{f_Y} \ar[dd]_{i_Y}
& & D Y \ar[r]^-{\eps_Y} \ar[dd]^{\inl}  
    & Y \ar@{=}[r]  \ar[dd]_-{\inl}
      & Y  \ar[dd]^{\inl} \\
\\
D Y \ar[r]^-{\de_Y}
& D D Y \ar[r]^-{c_{DY}}
  & D Y + D Y \ar[r]^-{\eps_Y + \eps_Y}
    & Y + Y \ar[r]^-{\id + \inr}
      & Y + (Y + Y) \ar@{=}[dd]\\
Q Y \ar[u]^{j_Y} \ar[rr]^-{g_Y} 
& & D Y \ar[u]_{\inr} \ar[d]^{\inr} \ar[r]^{\eps_Y} 
    & Y \ar[u]_{\inr} \ar[d]^{\inr} \\
D Y \ar[r]^-{\de_Y} \ar@{=}[d]
& D D Y \ar[r]^-{c_{DY}}
  & D Y + D Y \ar[r]^-{c_Y + \eps_Y}
    & (Y + Y) + Y \ar[r]^-{\alpha}
      & Y + (Y + Y) \ar@{=}[d] \\
D Y \ar[r]^-{\de_Y} 
& D D Y \ar[r]^-{c_{DY}}
  & D Y + D Y \ar[rr]^-{\eps_Y + c_Y}
    & & Y + (Y + Y) \\
Q Y \ar@/^2pc/@{=}[uuu] \ar[u]_{j_Y} \ar[rr]^-{g_Y}
 & & D Y \ar[u]_{\inr} \ar[rr]^-{c_Y}
    & & Y + Y \ar[u]_{\inr}
}
\]

The result now follows from noticing that
$(\eps_Y + \eps_Y) \comp c_{D Y} \comp \delta_Y = c_{Ð Y}$:
\[
\small
\xymatrix@R=1.5pc{
D Y \ar@{=}[dr] \ar[r]^{\de_Y} 
& D D Y \ar[d]^{D \eps_Y} \ar[r]^-{c_{D Y}}
  & D Y + D Y \ar[d]^{\eps_Y + \eps_Y} \\
& D Y \ar[r]^-{c_Y}
  & Y + Y
}
\]

Left corectangularity is proved analogously. 
\qed\end{proof}
\end{example}

\begin{example} \label{ex:update-d}
  Going back to Example~\ref{ex:update-c}, 
  the update monad $T X = A \fun (B \times X)$ with $B = (B, \o, \pl)$ a
monoid and $(A, \dn)$ a $B$-set arises as the quotient of the monad
$T_0 X = \mu Z.\, X + (A \fun Z) + (B \times Z$) by the following three
equations for its operations $c_X : (A \fun X) \to T_0 X$ and
$d_X : B \times X \to T_0 X$:
\[
\small
\xymatrix@R=1.5pc{
X \ar[d] \ar[rrrr]^-{\eta_X}
& & & & T_0 X \ar@{=}[d] \\
A \fun (1 \times X) \ar[r]^-{\id \fun (\o \times \id)}
& A \fun (B \times X) \ar[r]^-{\id \fun c_X}
  & A \fun T_0 X \ar[r]^-{d_{T_0X}}
     & T_0 T_0 X \ar[r]^-{\mu_X}
       & T_0 X
}
\]
\[
\small
\xymatrix@R=1.5pc{
B \times (B \times X) \ar[d] \ar[r]^-{\id \times d_X}
& B \times T_0 X \ar[r]^-{d_{T_0 X}}
  & T_0 T_0 X \ar[r]^-{\mu_X}
    & T_0 X \ar@{=}[d] \\
(B \times B) \times X \ar[r]^-{{\oplus} \times \id}
& B \times X \ar[rr]^-{d_X}
  & & T_0 X
}
\]
\[
\small
\xymatrix@R=1.5pc{
A \fun (B \times (A \fun X)) \ar[d]^{\id \fun (\id \times ({\dn} \fun \id))} \ar[r]^-{\id \fun (\id \times c_X)}
& A \fun (B \times T_0 X) \ar[r]^-{\id \fun d_{T_0X}}
  & A \fun T_0 T_0 X \ar[r]^-{c_{T_0T_0X}}
    & T_0 T_0 T_0 X \ar[r]^-{\mu^{(3)}_X}
      & T_0 X \ar@{=}[d] \\
A \fun (B \times ((A \times B) \fun X)) \ar[r]
& A \fun (B \times X) \ar[r]^-{\id \fun d_X}
  & A \fun T_0 X \ar[r]^-{c_{T_0X}}
    & T_0 T_0 X \ar[r]^-{\mu_X}
      & T_0 X 
}
\]
The monad $T_0 X$ is the free monad on the functor $F X = (A \fun X) +
B \times X$. The dual of $F$ is $G Y = (A \times Y) \times (B \fun
Y)$.  The Sweedler dual of the monad $T$ is the subcomonad of the
cofree comonad $\sd{T_0} Y = \nu W.\, Y \times (A \times W) \times (B
\fun W)$ on the functor $G$ resulting from imposing the following
coequations on its cooperations $c'_Y : \sd{T_0} Y \to A \times Y$ and
$d'_Y : \sd{T_0} Y \to B \fun Y$:
\[
\small
\xymatrix@R=1.5pc{
\sd{T_0} Y \ar@{=}[d] \ar[rrrr]^-{\eps_Y}
& & & & Y \\
\sd{T_0} Y \ar[r]^{\de_Y}
& \sd{T_0} \sd{T_0} Y \ar[r]^-{d'_{\sd{T_0}Y}}
  & A \times \sd{T_0} Y \ar[r]^-{\id \times c'_Y}
    & A \times (B \fun Y) \ar[r]^-{\id \times (\o \fun \id)}
      & A \times (1 \fun Y) \ar[u]
}
\]
\[
\small
\xymatrix@R=1.5pc{
\sd{T_0} Y \ar@{=}[d] \ar[r]^-{\de_X} 
& \sd{T_0} \sd{T_0} Y \ar[r]^-{c'_{ÐY}}
  & B \fun \sd{T_0} Y \ar[r]^-{\id \fun c'_Y}
    & B \fun (B \fun Y)  \\
\sd{T_0} Y \ar[rr]^-{c'_Y}
& & B \fun Y \ar[r]^-{{\pl} \fun \id}
    & (B \times B) \fun Y \ar[u]
}
\]
\[
\small
\xymatrix@R=1.5pc{
\sd{T_0} Y \ar@{=}[d] \ar[r]^-{\de^{(3)}_Y}
& \sd{T_0} \sd{T_0} \sd{T_0} Y \ar[r]^-{c'_{\sd{T_0}\sd{T_0}Y}}
  & A \times \sd{T_0} \sd{T_0} Y \ar[r]^-{\id \times d'_{\sd{T_0}Y}}
     & A \times (B \fun \sd{T_0} Y) \ar[r]^-{\id \times (\id \fun c'_Y)} 
       & A \times (B \fun (A \times Y)) \\
\sd{T_0} Y \ar[r]^{\de_Y}
& \sd{T_0} \sd{T_0} Y \ar[r]^-{c'_{ÐY}}
  & A \times \sd{T_0} Y \ar[r]^-{\id \times d'_Y}
    & A \times (B \fun Y) \ar[r]
      & A \times (B \fun ((A \times B) \times Y)) \ar[u]_{\id \times (\id \fun ({\dn} \times \id))}
}
\]
Calculating, we can find that $\sd{T} Y \cong A \times (B \fun Y)$.\footnote{Also this calculation appeared in \cite{Uus15}.}
The comonad map $i : \sd{T} \to \sd{T_0}$ is defined by 
$i\, (a, f) = (f\, \o, (a, i\, (a, f)), \lambda b.\, i\, (a \dn b, \lambda b'.\, 
f\, (b \pl b)))$.

\end{example}

\section{Related work}
\label{sec:related}

Works closest related to this paper on monad-comonad interaction laws
are Power and Shkaravska's work on arrays (lenses) as
comodels~\cite{PS04}, Power and Plotkin's study of tensors of models
and comodels~\cite{PP08}, Abou-Saleh and Pattinson's work on comodels
for operational semantics~\cite{ASP13}, M{\o}gelberg and Staton's work
on linear usage of state \cite{MS14} and Uustalu's work on
runners~\cite{Uus15}---the starting point for this work. Pattinson and
Schr\"oder~\cite{PS15} studied equational reasoning about comodels and
noted the degeneracy from nullary and binary cocommutative
cooperations; see also Bauer's tutorial \cite{Bau18}. Runners share
some features with Plotkin and Pretnar's algebraic effect
handlers~\cite{PP13}, we describe them in the end of this section.  In
their new work \cite{AB19}, Ahman and Bauer proposed a language
design for (residual) runners.

Hancock and Hyvernat's work on interaction structures \cite{HH06}
centers on the canonical interaction law of the free monad on $\d{F}$
and the cofree comonad on $F$ where $F$ is a container
functor. (Intuitionistic) linear-logic based two-party session typing
\cite{TCP11} is very much about canonical interaction between
syntactically dual functors, as we discuss in the end of this
section. The same idea is central in game-theoretic semantics of
(intuitionistic) linear logic (formulae-as-games,
proofs-as-strategies) \cite{AJ94}.
\tu{what is a good intuitionistic lin logic games semantics reference to put here, do we need one at all?}

\paragraph{Runners vs.\ handlers}

There are some similarities between handlers (now often called deep
handlers) to runners, but also significant differences.

Given a monad $T = (T, \eta, \mu)$ on $\C$. We are interested in
handling or running computations specified by $T$.

A \emph{handler} \cite{PP13} for an object (value set for input
computations) $X$ is mathematically an algebra of the monad $T$ on
$X$, i.e., an object $Z$ (return value set) with a map
$\alpha : T Z \to Z$ satisfying the conditions of a monad algebra,
that also comes with a map $f : X \to Z$.

A handler induces a unique map $h : T X \to Z$ satisfying 
\begin{equation} \label{eq:handler}
\small
\xymatrix@R=1.5pc@C=4pc{
X \ar[d]_{\eta_X} \ar[dr]^f \\
T X \ar@{.>}[r]^{h}
& Z \\
T T X \ar[u]^{\mu_X} \ar@{.>}[r]^{Th}
& T Z \ar[u]_{\alpha}
}
\end{equation}
as $((T X, \mu_X), \eta_X)$ is the free algebra of $T$ on $X$.

A runner, as we know, can be taken to be an object $Y$ (state set)
with a monad map from $T$ to the state monad
$\St^Y = (\St^Y, \eta^Y, \mu^Y)$, i.e., a natural transformation
$\vartheta : T \to \St^Y$ such that, for any $X$, we have
\[
\small
\xymatrix@R=1pc@C=4pc{
X \ar[dd]_{\eta_X} \ar[ddr]^{\eta^Y_X} \\ 
\\
T X \ar@{.>}[r]^{\vartheta_X}
& \St^Y X \\
& \St^Y \St^Y X \ar[u]_{\mu^Y_X} \\
T T X \ar[uu]^{\mu_X} \ar@{.>}[r]^{T\vartheta_X}
& T \St^Y X \ar[u]_{\vartheta_{\St^Y X}}
}
\]

We know that that such natural transformations $\vartheta$ are in a
bijection with coalgebras of the comonad $\sd{T}$ with carrier $Y$,
i.e., maps $\gamma : Y \to \sd{T} Y$ satisfying the conditions of a
comonad coalgebra.

We can see that a handler induces a map $h$ with domain $TX$ where $X$
an arbitrary fixed object; the codomain of $h$ can be anything---$Z$ is
an arbitrary fixed object. A runner, at the same time, is a family of
maps $\vartheta_X$ with domain $TX$ where $X$ can be varied to be any
object. The codomain of $\vartheta_X$ is of a prescribed form---it has
to be $\St^Y X$ where $Y$ is an arbitrary fixed object. The map $h$
is induced by an algebra of the monad $T$ while the family of maps
$\vartheta_X$ is induced by (and also induces) a coalgebra of the
comonad $\sd{T}$.

We can make this comparison fairer by acknowledging that algebras
$(Z, \alpha)$ of the monad $T$ are in a bijection with monad maps from
$T$ to the continuations monad $(\Cont^Z, \eta^Z, \mu^Z)$ for
the answer set $Z$, 
defined by $\Cont^Z X = (X \fun Z) \fun Z$.  Then
for any $X$ and $f : X \to Z$, the map $h$ from diagram
(\ref{eq:handler}) factorizes as
\[
\small
\xymatrix@R=1pc@C=4pc{
X \ar[dd]_{\eta_X} \ar[ddr]^{\eta^Z_X} \\
\\
T X \ar@{.>}[r]^{\xi_X}
& \Cont^Z X \ar[r]^{\lambda k.\, k\, f}  
  & Z \\
& \Cont^Z \Cont^Z X \ar[u]_{\mu^Z_X} \\
T T X \ar[uu]^{\mu_X} \ar@{.>}[r]^{T\xi_X}
& T \Cont^Z X \ar[u]_{\xi_{\Cont^Z X}}
}
\]
where $\xi : T \to \Cont^Z$ is the monad map corresponding to the
algebra structure $\alpha : T Z \to Z$.

Now both the handler-induced function $\xi$ and the runner $\vartheta$
are functions with domain $TX$ polymorphic in $X$. Still the
handler-induced function $\xi$ is specified by an algebra of the monad
$T$ while the runner $\vartheta$ is specified by a coalgebra of the
comonad $\sd{T}$.

Conceptually, handlers and runners/interaction laws are really
different in that, in the case of handlers, effects are treated inside
a computation while runners/interaction laws use an outside machine to
do this.

\paragraph{Session types}

In session type systems, one usually works with an inductively defined
set of types along the lines of 
\[
\renewcommand{\arraystretch}{1}
\begin{array}{rcl@{\quad}l}
G & := & Y & \textrm{return} \\
& \mid & G_0 + G_1 & \textrm{internal choice} \\ 
& \mid & G_0 \times G_1 & \textrm{external choice} \\ 
& \mid & A \times G_0 & \textrm{output} \\
& \mid & A \fun G_0 & \textrm{input}
\end{array}
\]
where $Y$ is a type variable and $A$ is a base type. (For simplicity,
we ignore inductive and coinductive types here.)  Internal choice and
external choice are really just special cases of output resp.\ input
for $A = \Bool$.

The dual of a type is defined recursively by
\[
\renewcommand{\arraystretch}{1}
\begin{array}{rcl}
\d{Y} & = & Y \\
\d{(G_0 + G_1)} & = & \d{G_0} \times \d{G_1} \\
\d{(G_0 \times G_1)} & = & \d{G_0} + \d{G_1} \\
\d{(A \times G_0)} & = & A \fun \d{G_0} \\
\d{(A \fun G_0)} & = & A \times \d{G_0}
\end{array}
\]

This syntactically defined dual agrees with our semantic concept of
the dual of a functor, except for the last clause where a discrepancy
arises. We work in an arbitrary Cartesian closed category.  In session
typing a linear setting is intended.

\section{Conclusion and future work}
\label{sec:conclusion}

We hope to have demonstrated that monad-comonad interaction laws are a
natural concept for describing interaction of effectful computations
with machines providing the effects. They are well-motivated not only
as a computational model, but also mathematically, admitting an
elegant theory based on concepts and methods that have previously
proved useful in other mathematical contexts, such as the Sweedler
dual.

There are many questions that we have not yet answered. What
are some general ways to compute the Sweedler dual? Power's work
\cite{PS04} suggests a sophisticated iterative construction based on
improving approximations. What is a good general syntax for
cooperations and coequations?  What can be said about the ``dual''
and the Sweedler ``dual'' in the presence of a residual monad and how
to compute them in this situation? How to compute the
Sweedler dual in some intuitionistic linear setting adequate for session typing?

\er{
Questions we might try to answer:
\begin{itemize}
\item What's the relation to handlers?
\item What's the relation to comodels?
\item What's the relation to distributive laws between monads and comonads?
\item What's the relation to Frobenius monoids?
\end{itemize}
}

\paragraph{Acknowledgements} 

We are grateful to Robin Cockett for discussions and encouragement
and to Ignacio L\'{o}pez Franco for pointing out the categorical work on
measuring morphisms.

T.U.\ was supported by the Icelandic Research Fund project grant
no.~196323-051, the Estonian Ministry of Education and Research
institutional research grant no.~IUT33-13, a project of the
Estonian-French Parrot cooperation programme and a guest professorship
from Universit\'e Paris 13. E.R.\ was in part supported by the
European Research Council starting grant no. 715753 (SECOMP) and by
Nomadic Labs via a grant on ``Evolution, Semantics, and Engineering of
the $F^*$ Verification System''.  E.R.\ also benefited from the
above-mentioned Parrot and Icelandic Research Fund projects.

\tu{any other acks?}\er{added my support + benefiting from Parrot and Icelandic Research Fund, is this last okay?}

%
%
%

\nocite{*}  \tu{remove the nocite command before submitting}

\bibliographystyle{splncs04}
\bibliography{preprint}
%

\end{document}